\documentclass[journal,onecolumn,11pt]{IEEEtran}
\usepackage{amsfonts,amsmath,amssymb}
\usepackage{indentfirst, setspace}
\usepackage{url,float}
\usepackage{color}
\usepackage[a4paper,ignoreall]{geometry}
\usepackage{longtable}
\usepackage{graphicx}
\usepackage[a4paper,ignoreall]{geometry}
\usepackage{multicol}
\usepackage{stfloats}
\usepackage{enumerate}
\usepackage{cite}
\usepackage{amsthm}
\usepackage{booktabs}
\usepackage{mathrsfs}
\usepackage{multirow}
\usepackage{amssymb}
\usepackage{verbatim}
\usepackage{cases}
\usepackage[square, comma, sort&compress, numbers]{natbib}
\usepackage[overload]{empheq}
\def\qu#1 {\fbox {\footnote {\ }}\ \footnotetext { From Qu: {\color{red}#1}}}
\def\hqu#1 {}
\def\kq#1 {\fbox {\footnote {\ }}\ \footnotetext { From KangQuan: {\color{blue}#1}}}
\def\hkq#1 {}

\usepackage{stfloats}

\newtheorem{Th}{Theorem}

\newtheorem{Prop}[Th]{Proposition}

\newtheorem{Lemma}[Th]{Lemma}
\newtheorem{Def}[Th]{Definition}

\newtheorem{Conj}[Th]{Conjecture}

\newcommand{\tr}{{\rm Tr}}

\newcommand{\gf}{{\mathbb F}}

\makeatletter
\newcommand{\figcaption}{\def\@captype{figure}\caption}
\newcommand{\tabcaption}{\def\@captype{table}\caption}
\makeatother

\begin{document}
	\title{Cryptographically Strong Permutations \\ from the Butterfly Structure}
	\author{{ Kangquan Li, Chunlei Li, Tor Helleseth and Longjiang Qu}
		\thanks{\noindent Kangquan Li is with the College of Liberal Arts and Sciences,
			National University of Defense Technology, Changsha, 410073, China and is currently a visiting Ph.D. student at the Department of Informatics, University of Bergen, Bergen N-5020, Norway.
			Chunlei Li and Tor Helleseth are with the Department of Informatics, University of Bergen, Bergen N-5020, Norway.
			Longjiang Qu is with the College of Liberal Arts and Sciences,
			National University of Defense Technology, Changsha, 410073, China, and is also with
			the State Key Laboratory of Cryptology, Beijing, 100878, China.
			
			\smallskip
			
			\textbf{Emails}: 			likangquan11@nudt.edu.cn,  chunlei.li@uib.no, tor.helleseth@uib.no, 
			ljqu\_happy@hotmail.com
		}
	}
	\maketitle{}

\begin{abstract}
%
    	In Crypto'2016 Perrin et al. discovered the butterfly structure that contains the Dillon APN permutation
    of six variables. The novel idea of this structure is the representation of 
    certain permutations of $\mathbb{F}_{2^{2n}}$ in terms of bivariate polynomials over $\mathbb{F}_{2^n}$.
    The butterfly structure was later generalized, which turns out to be a powerful approach that generates infinite families of cryptographic functions with best known nonlinearity and differential properties.  
    This motivates us to construct cryptographically strong permutations from generalized butterfly structures.
	
	Boomerang connectivity table (BCT)  is a new tool introduced by Cid et al. in Eurocrypt'18 to evaluate 
	the vulnerability of cryptographic functions against boomerang attacks. Consequently, a cryptographic function 
	is desired to have boomerang uniformity as low as its differential uniformity.
	Based on generalized butterfly structures, this paper presents infinite families of permutations of $\mathbb{F}_{2^{2n}}$ for a positive odd integer $n$, which have high nonlinearity $2^{2n-1}-2^{n}$ and boomerang uniformity $4$. 
	We investigated both open and closed butterfly structures. It appears, according to experiment results, that open butterflies 
	do not produce permutation with boomerang uniformity $4$.
    On the other hand, for the closed butterflies,  we obtain a condition on coefficients $\alpha, \beta \in \mathbb{F}_{2^n}$ such that the functions
    $$V_i(x,y) := (R_i(x,y), R_i(y,x)), $$ where $R_i(x,y)=(x+\alpha y)^{2^i+1}+\beta y^{2^i+1}$ and $\gcd(i,n)=1$, permute $\gf_{2^n}^2$ and have boomerang uniformity $4$. 
    The main result in this paper consists of two major parts: the permutation property is investigated in terms of the univariate form of $V_i$, and 
    the boomerang uniformity is examined in terms of 
    the original bivariate form. In addition, experiment results for $n=3, 5$ indicate that the proposed condition seems to cover all coefficients $\alpha, \beta \in \mathbb{F}_{2^n}$ that 
    produce permutations $V_i$ with boomerang uniformity $4$.  
    
    However, the experiment result shows that the quadratic permutation $V_i$ seems to be affine equivalent to the Gold function. Therefore, unluckily, we may not to obtain new permutations with boomerang uniformity $4$ from the butterfly structure. 
\end{abstract}


\section{Introduction}

Substitution boxes, known as S-boxes, are crucial nonlinear building blocks in modern block ciphers. 
In accordance with known attacks in the literature, Sboxes used in block ciphers are required to satisfy 
a variety of cryptographic criteria, including high nonlinearity \cite{chabaud1994links}, low differential uniformity \cite{nyberg1993differentially} and bijectivity.
In Eurocrypt'18, Cid et al. introduced a new tool of S-boxes, so-called the boomerang connectivity table (BCT), which similarly
analyzes the dependency between the upper part and lower part of a block cipher in a boomerang attack \cite{cid2018boomerang}. 
This new tool quickly attracted researchers' interest in studying properties and bounds of BCT of cryptographic functions.  Boura and Canteaut in \cite{boura2018boomerang} 
investigated the relation between entries in BCT and difference distribution table (DDT), and 
introduced the notion of the boomerang uniformity, which is the maximum value in BCT among all nonzero differences of inputs and outputs. 
They completely characterized the BCTs of $4$-bit S-boxes with differential uniformity $4$ classified in \cite{leander2007classification}, and
also determined the boomerang uniformities of the inverse function and the Gold function. Later Li et al. in \cite{li2019new} provided an equivalent formula to compute the boomerang uniformity of a cryptographic function.
Using the new formula, they characterized the boomerang uniformity 
by means of the Walsh transform, and computed the boomerang uniformities of some permutations with low differential uniformity. Mesnager et al.  considered the boomerang uniformity of quadratic permutations in \cite{mesnagerboomerang}, where they presented a characterization of quadratic permutations with boomerang uniformity $4$ and showed that the boomerang uniformity of certain quadratic permutations is preserved under extended affine (EA) equivalence.
Very recently, Calderini and Villa \cite{calderiniboomerang} also investigated the boomerang uniformities of some non-quadratic permutations with differential uniformity $4$.

It is shown that the boomerang uniformity of a cryptographic function is greater than or equal to its differential uniformity, and that   
the lowest possible boomerang uniformity $2$ is achieved by
almost perfect nonlinear (APN) functions \cite{cid2018boomerang,boura2018boomerang}. Clearly, APN permutations operating on even number of variables 
are of greatest interest, which is referred to as the \textit{BIG APN} problem in the community.
Nonetheless, by far no other instance for this problem, except for the Dillon APN permutation of $\mathbb{F}_{2^6}$, has been found. 
Therefore, it is particularly interesting to construct permutations of $\mathbb{F}_{2^{2n}}$ that have high nonlinearity, differential and boomerang uniformity $4$. 
Up to now, there are only three infinite and inequivalent families of permutations over $\gf_{2^{2n}}$ that have boomerang uniformity $4$ for odd integers $n\geq 1$: 
\begin{enumerate}[(1)]
	\item $f(x)=x^{-1}$ over $\gf_{2^{2n}}$\cite{boura2018boomerang};
	\item $f(x)=x^{2^{2i}+1}$ over $\gf_{2^n}$, where $\gcd(i,n)=1$ \cite{boura2018boomerang};
	\item $ f(x) = \alpha x^{2^{2s}+1}+\alpha^{2^{2k}}x^{2^{-{2k}}+2^{2k+2s}}$  over $\gf_{2^{2n}}$, where $n=3k$, $3\nmid {k}$, $3 \mid {(k+s)}$, $\gcd(3k,s)=1$,  and $\alpha$ is a primitive element of $\gf_{2^{2n}}$ \cite{mesnagerboomerang}.
\end{enumerate}

 In Crypto'16, Perrin et al. investigated the only APN permutation over $\gf_{2^6}$  \cite{browning2010apn} by means of reverse-engineering and proposed the open butterfly and the closed butterfly structures \cite{perrin2016cryptanalysis}. A generalized butterfly structure was later proposed by Canteaut et al. \cite{canteaut2017generalisation}. The butterfly structures represent functions over $\gf_{2^n}^2$ in terms of bivariate form. 
 It is shown that the open butterfly structure produces permutations of $\gf_{2^n}^2$, which are CCZ-equivalent to the functions with simpler forms derived from the closed structure \cite{perrin2016cryptanalysis}.
Since differential uniformity is an invariant under CCZ-equivalence, one may consider to combine open and closed butterfly structures to construct permutations with low differential uniformity.
 As a matter of fact, by investigating differential uniformity of functions  from the closed butterfly structure, 
 researchers constructed several infinite families of differentially 4-uniform permutations  over $\gf_{2^n}^2$ with the open butterfly structure \cite{li2018generalization,fu2017differentially,canteaut2017generalisation,canteaut2018if}. 
Motivated by recent works on the butterfly structure, this paper aims to construct
infinite families of permutations with boomerang uniformity $4$ from generalized butterfly structures.  

Let $n$ be a positive odd integer and $q=2^n$.
Let  $\gamma$ be a primitive element of $\mathbb{F}_{2^2}$, i.e, $\gamma^2=\gamma + 1$. 
Since $n$ is odd,  the finite field $\gf_{q^2}=\gf_q(\gamma)$ and the basis $1, \gamma$ of $\mathbb{F}_{q^2}$ over $\mathbb{F}_q$ induces
a one-to-one correspondence between $\mathbb{F}_q^2$ and $\mathbb{F}_{q^2}$ as follows:  
$$
z = x+\gamma y   \leftrightarrow
 (x, y) = (\gamma^2z+\gamma z^q, z^q+z).
$$
Define a bivariate polynomial 
$$R_i(x,y)=(x+\alpha y)^{2^i+1}+\beta y^{2^i+1}, \quad \alpha,\beta\in\gf_{q}.$$ 
Since $n$ is odd, it is clear that the mapping $x\mapsto R_i(x,y)$ is a permutation of $\mathbb{F}_q$ for any fixed $y\in \mathbb{F}_q$.
From experiment results, the open butterfly structure based on $R_i(x,y)$ given in \cite{canteaut2017generalisation}
 seems not to yield  permutation with boomerang uniformity $4$ of $\gf_{2^3}^2$. Therefore, we will concentrate on the closed butterfly structure. 
Recall that the closed butterfly structure of $\gf_{q}^2$ from $R_i(x,y)$ is given by $$V_i(x,y) = (R_i(x,y), R_i(y,x)).$$  
According to the one-to-one correspondence between $\mathbb{F}_q^2$ and $\mathbb{F}_{q^2}$, the closed butterfly structure $V_i(x,y)$ over $\mathbb{F}_q^2$ can be expressed in a univariate form as $R_i(x,y) + \gamma R_i(y,x)$ with $z=x+\gamma y$ and $(x, y) = (\gamma^2z+\gamma z^q, z^q+z)$. 
By substituting $z$ with $\gamma z$ when $i$ is odd (resp. $\gamma^2 z$ when $i$ is even), this univariate polynomial can be transformed as 
\begin{equation}
\label{polynomial_representation}
f(z)=\epsilon_1z^{q\cdot(2^i+1)}+ \epsilon_2 z^{q\cdot 2^i+1} +  \epsilon_3z^{2^i+q} +  \epsilon_4 z^{2^i+1},
\end{equation}
where the coefficients satisfy 
$(\epsilon_1, \epsilon_2, \epsilon_3, \epsilon_4) = (\varepsilon_1, \varepsilon_2, \varepsilon_3, \varepsilon_4)$ for even $i$; and 
$(\epsilon_1, \epsilon_2, \epsilon_3, \epsilon_4) = (\varepsilon_3, \varepsilon_4, \varepsilon_1, \varepsilon_2)$ for odd $i$  with
 \begin{equation}
 \label{ieven}
\left\{
\begin{array}{lr}
 \varepsilon_1 =  \alpha^{2^i}+\alpha+1\\
  \varepsilon_2=\alpha^{2^i+1}+\alpha+\beta+1 \\
 \varepsilon_3=\alpha^{2^i+1}+\alpha^{2^i}+\beta+1\\
  \varepsilon_4=\alpha^{2^i+1}+\alpha^{2^i}+\alpha+\beta.
\end{array}
\right.
\end{equation} 
For the coefficients $\epsilon_1, \epsilon_2, \epsilon_3, \epsilon_4$ in the polynomial $f_i(z)$, define 
	  \begin{equation}	\label{varphi}
\left\{
\begin{array}{lr}
\varphi_1=\epsilon_1\epsilon_3+\epsilon_2\epsilon_4  \\ 
\varphi_2=\epsilon_1\epsilon_2+\epsilon_3\epsilon_4 \\
\varphi_3=\epsilon_1^2+\epsilon_4^2 \\
\varphi_4=\epsilon_1^2+\epsilon_2^2+\epsilon_3^2+\epsilon_4^2.
\end{array}
\right.
\end{equation} 
 This paper aims to characterize the condition on $\alpha$ and $\beta$ such that $V_i$ permutes $\gf_{q}^2$ and has boomerang uniformity $4$, and 
 the main result is given as follows.

%

\begin{Th}
	\label{main_theorem}  Let $q=2^n$ with $n$ odd, $\gcd(i,n)=1$ and $R_i(x,y)=(x+\alpha y)^{2^i+1}+\beta y^{2^i+1}$ with $\alpha, \beta\in \mathbb{F}_q^*$, where $\mathbb{F}_q^* = \mathbb{F}_q\setminus \{0\}$.
	Then the function $V_i(x,y) := (R_i(x,y), R_i(y,x))$ permutes $\gf_q^2$ and has boomerang uniformity  $4$ if $(\alpha, \beta)$ is taken from the 
	following set 
	\begin{equation}
	\label{Gamma}
	\Gamma = \left\{ (\alpha,\beta)\in\gf_{q}^*\times \gf_{q}^* ~\big|~ \varphi_2^{2^i}=\varphi_1\varphi_4^{2^i-1} ~\text{and}~\varphi_4\neq0  \right\},
	\end{equation}
	where  $\varphi_1, \varphi_2, \varphi_4$ are given in \eqref{varphi}.
\end{Th}

According to Mesnager et al.'s resut in \cite{mesnagerboomerang}, a quadratic permutation of $\mathbb{F}_{4^m}$ in the form
		\begin{equation}
		\label{special_form}f(z)=\sum_{0\le j\le k\le m-1}c_{ij}z^{4^j+4^k}, \quad c_{jk}\in\gf_{4^m},
		\end{equation}  has boomerang uniformity $4$ if its differential uniformity equals $4$. The function in Theorem \ref{main_theorem} has the univariate polynomial $f_i(z)=\epsilon_1z^{q\cdot(2^i+1)}+ \epsilon_2 z^{q\cdot 2^i+1} +  \epsilon_3z^{2^i+q} +  \epsilon_4 z^{2^i+1}$ with $q=2^n$ and $n$ odd. Clearly, $q\cdot(2^i+1)$ and $2^i+1$ can not belong to the same cyclotomic class of $4^j+4^k$ for any integers $j,k$ since $n$ is odd. Hence, it seems that our results can not be reduced by the above result. 

The rest of this paper is organized as follows:
we firstly recall the definitions of differential uniformity, boomerang uniformity, butterfly structure and introduce some useful lemmas in Section \ref{Preliminaries}. Sections \ref{permutation} and \ref{boomerang} are devoted to 
proving the permutation property and the boomerang uniformity in Theorem \ref{main_theorem}, respectively. Finally,  Section \ref{Conclusions} draws a conclusion of our work and raises some open questions. 

\section{Preliminaries}
\label{Preliminaries}
In this section, we assume $n$ is an arbitrary positive integer and $q=2^n$. Let $\tr_q(\cdot)$ denote the absolute trace function over $\gf_q$, i.e., $\tr_q(x) = x+x^2+\cdots+x^{2^{n-1}}$ for any $x\in\gf_q$. For any set $E$, the nonzero elements of $E$ is denoted by $E\backslash\{0\}$ or $E^{*}$. 

\subsection{Differential Uniformity and Boomerang Uniformity}

The concept of differential uniformity was introduced to reveal the subtleties of differential attacks.

\begin{Def} \cite{nyberg1993differentially}
	Let $f(x)$ be a function from $\gf_q$ to itself and $a,b\in\gf_q$. 
	The difference distribution table (DDT) of $f(x)$ is given by a $q \times q$ table $D$, in which the entry for the $(a,b)$ position is given by 
	$$DDT(a,b)=\#\{x\in\gf_q | f(x+a)+f(x) =b  \}. $$ 
The differential uniformity  of $f(x)$ is given by
	$$\Delta_f = \max \limits_{a\in \mathbb{F}_q^*, b\in\gf_{q}} DDT(a,b).  $$
\end{Def}
It is straightforward for any function from $\mathbb{F}_q$ to itself, each entry in its DDT takes even value and its differential uniformity is no less than $2$. A function with the minimum possible differential uniformity $2$ is called an almost perfect nonlinear (APN) function.

In \cite{cid2018boomerang}, Cid et al. introduced the concept of boomerang connectivity table of a permutation $f$ from $\gf_{2}^n$ to itself { as follows, which is also suitable for the case $\gf_{2^n}$ clearly. Later, Boura and Canteaut introduced the concept of the boomerang uniformity, which is defiend by the maximum value in BCT excluding the first row and column. 
	
	\begin{Def}
		\cite{cid2018boomerang,boura2018boomerang}
		\label{BCT-def-Fq}
		Let $f$ be an invertible function from $\gf_q$ to itself  and $a,b\in\gf_q$. The boomerang connectivity table (BCT) of $f$ is given by a $q\times q$ table,  in which the entry for the $(a,b)$ position is given by 
		\begin{equation}
		\label{T(a,b)}
			BCT(a,b)=\#\left\{x\in\gf_q:  
		f^{-1}(f(x)+b)+f^{-1}(f(x+a)+b) = a \right\}.
		\end{equation}		
		{The boomerang uniformity}  of $f$ is defined by 
		{	$$\delta_f = \max \limits_{a,b\in\gf_q^*} BCT(a,b).  $$}
	\end{Def}
It is shown in \cite{cid2018boomerang,boura2018boomerang} that $BCT(a, b) \geq DDT(a, b)$ for any $a, b$ in $\mathbb{F}_q$. In \cite{li2019new}, Li et al. presented an equivalent formula to compute BCT and the boomerang uniformity without knowing $f^{-1}(x)$ and $f(x)$ simultaneously as follows.

\begin{Prop}
	\cite{li2019new}
	\label{boomerang-def}
	Let $q=2^n$ and $f(x)\in\gf_q[x]$ be a permutation polynomial over $\gf_q$. Then the BCT of $f(x)$ can be given by a $q\times q$ table $BCT$, in which the entry $BCT(a,b)$ for the $(a,b)$ position is given by the number of solutions $(x,y)$ in $\gf_q\times \gf_q$ of the following equation system.
	   \begin{equation}
	\label{boomerang-uniformity}
	\left\{
	\begin{array}{lr}
	f(x+a) + f(y+a) = b,  \\
	f(x) + f(y) = b. 
	\end{array}
	\right.
	\end{equation} 
 {Equivalently},  the boomerang uniformity of $f(x)$, given by $\delta_f$, is the maximum number of solutions in $\gf_q\times\gf_q$ of (\ref{boomerang-uniformity}) as $a,\,b$ run through $\gf_{q}^{*}$.
\end{Prop}

Let $f$ be a quadratic function  from $\gf_q$ to itself with $f(0)=0$. The associated symmetric bilinear mapping is given by $S_f(x,y)=f(x+y)+f(x)+f(y)$, where $x,\, y\in\gf_q.$ 
For any $a\in\gf_{q}$, define 
$$\mathrm{Im}_{f,a} = \{S_f(a,x): x \in\gf_{q}  \}.$$ Very recently, Mesnager et al. \cite{mesnagerboomerang} presented a characterization about quadratic permutations with boomerang uniformity  $4$ using the new formula (\ref{boomerang-uniformity}).
\begin{Lemma}
	\cite{mesnagerboomerang}
	\label{quadratic_boomerang}
	Let $q=2^n$ and $f$ be a quadratic permutation of $\gf_{q}$ with differential uniformity  $4$. Then the boomerang uniformity of $f$ equals $4$ if and only if $\mathrm{Im}_{f,a}=\mathrm{Im}_{f,b}$ for any $a,b\in\gf_{q}^{*}$ satisfying $S_f(a,b)=0$. 
\end{Lemma}

\subsection{The Butterfly Structure}

In Crypto'16, Perrin et al. \cite{perrin2016cryptanalysis} analyzed the only known APN permutation over $\gf_{2^6}$ \cite{browning2010apn} and discovered that the APN permutation over $\gf_{2^6}$ has a simple decomposition relying on $x^3$ over $\gf_{2^3}$. Based on the power permutation $x^e$ over $\gf_{2^n}$, they presented the open butterfly structure and the closed butterfly structure, which were later generalized by  Canteaut et al. in \cite{canteaut2017generalisation}.

\begin{Def}
	\cite{perrin2016cryptanalysis}
	Let $q=2^n$ and  $\alpha\in\gf_{q}$, $e$ be an integer such that $x^e$ is a permutation over $\gf_{q}$ and $R_k[e,\alpha]$ be the keyed permutation 
	$$R_{k}[e,\alpha](x) = (x+\alpha k)^e+k^e.$$
	The following functions 
	$$ H_e^{\alpha}(x,y) = \left( R_{R_y[e,\alpha](x)}^{-1}(y), R_y[e,\alpha](x)  \right), $$
	$$V_e^{\alpha}(x,y) = \left( R_y[e,\alpha](x), R_x[e,\alpha](y)  \right)$$
	are called the open butterfly structure and closed butterfly structure respectively. 
\end{Def}


\begin{Def}
	\cite{canteaut2017generalisation}
	Let $q=2^n$ and $R(x,y)$ be a bivariate polynomial of $\gf_{q}$ such that $R_y: x\to R(x,y)$ is a permutation of $\gf_{q}$ for all $y$ in $\gf_{q}$. The closed butterfly $V_R$ is the function of $\gf_{q}^2$ defined by 
	$$V_R(x,y)=(R(x,y),R(y,x)),$$
	and the open butterfly $H_R$ is the permutation of $\gf_{q}^2$ defined by 
	$$H_R(x,y)= \left( R_{R_y^{-1}(x)}(y), R_y^{-1}(x)  \right),$$
	where $R_y(x)=R(x,y)$ and $R_{y}^{-1}(R_y(x))=x$ for any $x,y$. 
\end{Def}

In \cite{li2018generalization}, Li et al. discussed the cryptographic properties, including the differential uniformity, the nonlinearity and the algebraic degree, of the butterfly structure of the form
$$V_i := (R_i(x,y), R_i(y,x))$$ with $R_i(x,y)=(x+\alpha y)^{2^i+1}+\beta y^{2^i+1}$.

\begin{Lemma}
	\label{differential_uniformity_butterfly}
	\cite{li2018generalization}
	Let $n$ be odd, $q=2^n$, $i$ be an integer with $\gcd(i,n)=1,$ $ \alpha,\beta \in\gf_{q}^{*}$ and $\beta\neq (\alpha+1)^{2^i+1}$. Then the differential uniformity of 
	$$V_i := (R_i(x,y), R_i(y,x))$$ 
	is at most $4$, where $R_i(x,y)=(x+\alpha y)^{2^i+1}+\beta y^{2^i+1}$.
\end{Lemma}
Since the boomerang uniformity of a function is no less than its differential uniformity, functions with differential uniformity $4$ is a natural starting point for constructing permutations with boomerang uniformity $4$.
In fact, the condition $\beta\neq (\alpha+1)^{2^i+1}$ in Lemma \ref{differential_uniformity_butterfly} corresponds to the condition $\varphi_4\neq 0$ in the set $\Gamma$ required in Theorem \ref{main_theorem}.

\subsection{Useful Lemmas}

  This subsection summarizes some lemmas that will be used for proving the permutation property of the function in Theorem \ref{main_theorem}.

\begin{Lemma}
	\label{lem1}
	\cite{park2001permutation,wang2007cyclotomic,zieve2009some}
	Pick $d,r > 0$ with $d\mid (q-1)$, and let $h(x)\in\gf_q[x]$. Then $f(x)=x^rh\left(x^{\left.(q-1)\middle/d\right.}\right)$ permutes $\mathbb{F}_q$ if and only if both
	\begin{enumerate}[(1)]
		\item $\gcd(r,\left.(q-1)\middle/d\right.)=1$ and
		\item $g(x)=x^rh(x)^{\left.(q-1)\middle/d\right.}$ permutes $\mu_d$, where $\mu_d=\{x\in\gf_q : x^d=1\}$.
	\end{enumerate}
\end{Lemma}

Let the unit circle of $\gf_{q^2}$ be defined by $$\mu_{q+1}:= \{ x\in\gf_{q^2} : x^{q+1}=1 \}.$$  The unit circle of $\gf_{q^2}$ has the following relation with the finite field $\gf_{q}$.
\begin{Lemma}
	\cite{lahtonen1995odd}
	\label{mu_Fq}
	Let $\gamma$ be any fixed element in $\gf_{q^2}\backslash\gf_{q}$. Then we have 
	$$\mu_{q+1}\backslash\{1\} = \left\{  \frac{x+\gamma}{x+\gamma^q}: x \in \gf_{q}   \right\}.$$
\end{Lemma}

The following lemma is about the solutions of a linear equation. The proof is easy and we omit it.
\begin{Lemma}
	\label{solutions_linear}
	Let $q=2^n$ and $\gcd(i,n)=1$. Then for any $a\in\gf_{q}$, the equation $x^{2^i}+x=a$ has solutions in $\gf_{q}$ if and only if $\tr_q(a)=0$. Moreover, when $\tr_q(a)=0$, the equation $x^{2^i}+x=a$ has exactly two solutions $x=x_0, x_0+1$  in $\gf_{q}$.
\end{Lemma}

\begin{Lemma}
	\cite{lidl1993dickson}
	\label{Dickson_property} Let $\mathbb{R}$ be a commutative ring with identity. The Dickson polynomial $D_k(x,a)$ of the first kind of degree $k$
	$$D_k(x,a)=\sum_{j=0}^{\lfloor \frac{k}{2} \rfloor}\frac{k}{k-j}\begin{pmatrix}
	k-j \\
	j
	\end{pmatrix}(-a)^jx^{k-2j}$$
	has the following properties:
	\begin{enumerate}[(1)]
		\item $D_k(x_1+x_2,x_1x_2)=x_1^k+x_2^k$, where $x_1,\,x_2$ are two indeterminates;
		\item $D_{k+2}(x,a)=xD_{k+1}(x,a)-aD_k(x,a)$;
		\item $D_{k\ell}(x,a)=D_{k}\left( D_{\ell}(x,a), a^{\ell} \right)$;
		\item if $\mathbb{R}=\gf_{2^n}$, then $D_{2^i}(x,a)=x^{2^i}$.
	\end{enumerate}
\end{Lemma}

By the above lemma, the Dickson polynomial of degree $k=2^i-1$ over $\gf_{2^n}$ can be explicitly given.  

\begin{Lemma}
	\label{Dickson_2i-1}For any positive integer $i$ and element $a\in\gf_{2^n}$, 
	\begin{equation}
	\label{Dickson}D_{2^i-1}(x,a)=\sum_{j=0}^{i-1}a^{2^{j}-1}x^{2^i-2^{j+1}+1}.
	\end{equation}
\end{Lemma}

\begin{proof}
	We prove the statement by induction. It is clear that (\ref{Dickson}) holds for $i=1$ since $D_1(x,a)= x$. Suppose that (\ref{Dickson}) holds for $i-1$, namely, 
		\begin{equation}
	\label{Dickson_i-1}D_{2^{i-1}-1}(x,a)=\sum_{j=0}^{i-2}a^{2^{j}-1}x^{2^{i-1}-2^{j+1}+1}.
	\end{equation}
	By Lemma \ref{Dickson_property} (3) and (4), we have 
	\begin{eqnarray*}
	D_{2^i-2}(x,a) &=& D_{{2^{i-1}-1}} \left( D_2(x,a) ,a^2\right) \\
	&=&  D_{{2^{i-1}-1}} \left( x^2 ,a^2\right) \\
	&=& \sum_{j=0}^{i-2}a^{2(2^{j}-1)}x^{2(2^{i-1}-2^{j+1}+1)}\\
	&=& \sum_{j=1}^{i-1}a^{2^{j}-2}x^{2^{i}-2^{j+1}+2}.
	\end{eqnarray*}
In addition, according to the second item of Lemma \ref{Dickson_property}, 
$$D_{2^i}(x,a) = xD_{2^i-1}(x,a)+aD_{2^i-2}(x,a).$$
Thus, 
\begin{eqnarray*}
D_{2^i-1}(x,a) &=& x^{-1}\left( x^{2^i} +  aD_{2^i-2}(x,a) \right) \\
&=& x^{-1}\left( x^{2^i} + a\sum_{j=1}^{i-1}a^{2^{j}-2}x^{2^{i}-2^{j+1}+2} \right) \\
&=& \sum_{j=0}^{i-1}a^{2^{j}-1}x^{2^i-2^{j+1}+1},
\end{eqnarray*}
which implies that (\ref{Dickson}) holds for the $i$ case. 
Therefore, the desired conclusion follows. 
\end{proof}

In the end of this section, we provide a lemma about some properties of the elements in $\Gamma$ defined by (\ref{Gamma}), which will be heavily used in the proof of the main theorem.

\begin{Lemma}
	\label{property of Gamma}
	Let $q=2^n$ with $n$ odd and $\gcd(i,n)=1$.	
	For any element $(\alpha,\beta)$ in $\Gamma$ in (\ref{Gamma}), the elements $\varphi_1, \varphi_2, \varphi_3, \varphi_4$ defined by (\ref{varphi}) satisfy the following properties:
	\begin{enumerate}[(1)]
		\item $(\varphi_1+\varphi_4)(\varphi_2+\varphi_4)(\varphi_3+\varphi_4)\varphi_3\neq0$, $\left(\frac{\varphi_4}{\varphi_2+\varphi_4}\right)^{2^i} = \frac{\varphi_4}{\varphi_1+\varphi_4} $ and $\left( \frac{\varphi_1+\varphi_4}{\varphi_2+\varphi_4} \right)^{\frac{1}{2^i-1}}=\frac{\varphi_2+\varphi_4}{\varphi_4}$;
		\item when $i$ is even, $\tr_q\left(\frac{\varphi_3}{\varphi_4}\right)=0$; moreover, the equation $$x^{2^i}+x+\frac{\varphi_3+\varphi_4}{\varphi_4}=0$$ have two solutions  $\frac{\varphi_2+\varphi_4}{\varphi_4}\alpha$ and $\frac{\varphi_2+\varphi_4}{\varphi_4}\alpha+1$ in $\gf_q$;
		\item when $i$ is odd, $\tr_q\left(\frac{\varphi_3}{\varphi_4}\right)=1$;
		\item $\tr_q\left(\frac{\varphi_2}{\varphi_4}\right)=0$.
	\end{enumerate}
\end{Lemma} 

\begin{proof}
	From the definition of $\varphi_1, \varphi_2, \varphi_3, \varphi_4$ in \eqref{varphi}, it follows that
	\begin{equation}
	\label{varphi_even_long}	
	\left\{
	\begin{array}{lr}
	\varphi_1= \alpha^{2^{i+1}+2} + \alpha^{2^{i+1}}+\alpha^{2^i+2}+\alpha^{2^i}+\alpha^2 +\alpha\beta + \beta^2+1  \\ 
	\varphi_2=  \alpha^{2^{i+1}+2} + \alpha^{2^{i+1}+1}+\alpha^{2^{i+1}}+\alpha^{2}+\alpha +\alpha^{2^i}\beta + \beta^2+1 \\
	\varphi_4=  \alpha^{2^{i+1}+2} + \alpha^{2^{i+1}}+ \alpha^2 + \beta^2+1.
	\end{array}
	\right.
	\end{equation} 
	and 
	\begin{equation}
	\label{varphi_3}	
	\varphi_3 =\begin{cases}
	\alpha^{2^{i+1}+2} + \beta^2+1 &\text{ for even } i \\
	\alpha^{2^{i+1}} + \alpha^2  & \text{ for odd } i.
	\end{cases}
	\end{equation} 
	The above equations imply \begin{equation}\label{Eq-vaphi1+2}
	\varphi_1+\varphi_2 = \alpha^{2^i+2}+\alpha^{2^i}+\alpha\beta + \alpha^{2^{i+1}+1}+\alpha^{2^i}\beta+\alpha=(\alpha^{2^i}+\alpha)( \alpha^{2^i+1}+\beta+1)
	\end{equation}
	and 
	\begin{equation}\label{Eq-varphi3+4}
	\Big\{\varphi_3, \varphi_3+\varphi_4\Big\} = \left\{(\alpha^{2^i}+\alpha)^2, ( \alpha^{2^i+1}+\beta+1 )^2\right\}.	
	\end{equation}
	
	(1)
	Recall from the definition of $\Gamma$ that $\varphi_2^{2^i}=\varphi_1\varphi_4^{2^i-1}$. Suppose $(\varphi_1+\varphi_4)(\varphi_2+\varphi_4)=0$. Then it is clear that $\varphi_1=\varphi_2=\varphi_4$. 
	Thus $\beta=\alpha^{2^i+1}+1$ or $\alpha^{2^i}+\alpha=0$. In fact, if $\beta=\alpha^{2^i+1}+1$,  then $\varphi_1+\varphi_4=\alpha^{2^i+2}+\alpha^{2^i}+\alpha\beta=\alpha^{2^i}+\alpha =0$.  Thus we always have $\alpha^{2^i}+\alpha=0$. This implies $\varphi_1+\varphi_4=\alpha^{2^i+2}+\alpha^{2^i}+\alpha\beta = \alpha\beta = 0$, which is in contradiction with the assumption that $\alpha\beta\neq0$ in the definition of $\Gamma$.
	Suppose $\varphi_3(\varphi_3+\varphi_4)=0$, we obtain $(\alpha^{2^i}+\alpha)(\alpha^{2^i}+\alpha)=0$. This leads to the contradiction $\alpha\beta =0$ in an exactly similar manner.
	The equalities 
	$$
	\left(\frac{\varphi_4}{\varphi_2+\varphi_4}\right)^{2^i} = \frac{\varphi_4}{\varphi_1+\varphi_4} \text{ and } \left( \frac{\varphi_1+\varphi_4}{\varphi_2+\varphi_4}\right)^{\frac{1}{2^i-1}}=\frac{\varphi_2+\varphi_4}{\varphi_4}
	$$
	can be easily verified by the relation $\varphi_2^{2^i} = \varphi_1\varphi_4^{2^i-1}$ in the definition of $\Gamma$.
	
	
	
	(2) From \eqref{varphi_even_long} and \eqref{varphi_3}, we have
	\begin{subequations} 
		\renewcommand\theequation{\theparentequation.\arabic{equation}}  
		\label{varphi_1+2+3+4}   
		\begin{empheq}[left={\empheqlbrace\,}]{align}
		&  \varphi_1+\varphi_4 = \alpha^{2^i+2}+\alpha^{2^i}+\alpha\beta  \label{lemma_eq_1} \\ 
		& \varphi_2+\varphi_4 = \alpha^{2^{i+1}+1}+\alpha^{2^i}\beta+\alpha \label{lemma_eq_2}  \\
		& \varphi_3+\varphi_4 = \alpha^{2^{i+1}}+\alpha^2. \label{lemma_eq_3} 
		\end{empheq}
	\end{subequations}
	From the above equalities, it is easy to verify that 
	\begin{equation}
	\label{lemma_eq_5}
	\alpha\left( \varphi_2+\varphi_4\right)+ \alpha^{2^i}\left(\varphi_1+\varphi_4 \right)=\varphi_3+\varphi_4.
	\end{equation}
	Moreover, using $\left(\frac{\varphi_4}{ \varphi_2+\varphi_4}\right)^{2^i} = \frac{\varphi_4}{\varphi_1+\varphi_4} $, we have
	\begin{equation}
	\label{v3v4}
	\frac{\varphi_3+\varphi_4}{\varphi_4}=\frac{\varphi_2+\varphi_4}{\varphi_4}\alpha+ \frac{\varphi_1+\varphi_4}{\varphi_4}\alpha^{2^i}= \frac{\varphi_2+\varphi_4}{\varphi_4}\alpha+ \left( \frac{\varphi_2+\varphi_4}{\varphi_4}\alpha \right)^{2^i}.
	\end{equation}
	Thus, 
	$$\tr_q\left(\frac{\varphi_3+\varphi_4}{\varphi_4}\right)=\tr_q\left(\frac{\varphi_3}{\varphi_4}\right)+\tr_q(1)=0.$$ Furthermore, from Lemma \ref{solutions_linear},  the solutions in $\gf_{q}$ of $x^{2^i}+x=\frac{\varphi_3+\varphi_4}{\varphi_4}$ are  $\frac{\varphi_2+\varphi_4}{\varphi_4}\alpha$ and $\frac{\varphi_2+\varphi_4}{\varphi_4}\alpha+1$.

	(3) From the expressions of $\varphi_3, \varphi_4$ in \eqref{varphi_even_long}, \eqref{varphi_3}, it is easily seen that 
	$\varphi_4 = \varphi_{3,e}+\varphi_{3,o}$, where $\varphi_{3,e}$, $\varphi_{3,o}$ denotes $\varphi_3$ for even $i$ and for odd $i$, respectively.
	Since $n$ is an odd integer, we have
	$$
	\tr_q\left(\frac{\varphi_{3,o}}{\varphi_4}\right) +\tr_q\left(\frac{\varphi_{3,e}}{\varphi_4}\right) = \tr_q(1) = 1.
	$$
	The desired assertion follows from the fact $\tr_q\left(\frac{\varphi_{3,e}}{\varphi_4}\right)=0$ proved in (2).

	(4)  From \eqref{Eq-vaphi1+2} and \eqref{Eq-varphi3+4}, it is easily seen that
	\begin{equation*}
	\varphi_1^2+\varphi_2^2=\varphi_3^2+\varphi_3\varphi_4,
	\end{equation*}
	i.e.,
	\begin{equation}
	\label{lemma_eq_6}
	\left(\frac{\varphi_1}{\varphi_4}\right)^2+\left(\frac{\varphi_2}{\varphi_4}\right)^2=\left(\frac{\varphi_3}{\varphi_4}\right)^2+\frac{\varphi_3}{\varphi_4}.
	\end{equation}
	Plugging $\frac{\varphi_1}{\varphi_4} = \left(\frac{\varphi_2}{\varphi_4}\right)^{2^i}$ into Eq. (\ref{lemma_eq_6}), we get 
	$$ \left(\frac{\varphi_2}{\varphi_4}\right)^{2^{i+1}} +  \left(\frac{\varphi_2}{\varphi_4}\right)^2=\left(\frac{\varphi_3}{\varphi_4}\right)^2+\frac{\varphi_3}{\varphi_4}=\left(\frac{\varphi_3+\varphi_4}{\varphi_4}\right)^2+\frac{\varphi_3+\varphi_4}{\varphi_4}. $$
	By the relation between $\varphi_4$ and $\varphi_3$ for even and odd $i$, it is clear that the expression on the right side of the above equation is independent of the parity of the integer $i$.
	W.L.O.G., we can assume that $i$ is even, since the case $i$ odd can be proved by just replacing $\varphi_3$ by $\varphi_3+\varphi_4$. 
	Together with (\ref{v3v4}), we have 
	$$ \left(\frac{\varphi_2}{\varphi_4}\right)^{2^{i+1}} +  \left(\frac{\varphi_2}{\varphi_4}\right)^2=\left(  \frac{\varphi_2+\varphi_4}{\varphi_4}\alpha+\left( \frac{\varphi_2+\varphi_4}{\varphi_4}\alpha\right)^2 \right)^{2^i} + \frac{\varphi_2+\varphi_4}{\varphi_4}\alpha+\left( \frac{\varphi_2+\varphi_4}{\varphi_4}\alpha\right)^2 . $$ 
	Therefore, $$ \left(\frac{\varphi_2}{\varphi_4}\right)^2 = \frac{\varphi_2+\varphi_4}{\varphi_4}\alpha+\left( \frac{\varphi_2+\varphi_4}{\varphi_4}\alpha\right)^2$$
	or 
	\begin{equation}
	\label{lemma_eq_7} \left(\frac{\varphi_2}{\varphi_4}\right)^2 = \frac{\varphi_2+\varphi_4}{\varphi_4}\alpha+\left( \frac{\varphi_2+\varphi_4}{\varphi_4}\alpha\right)^2+1.
	\end{equation}
	
	If Eq. (\ref{lemma_eq_7}) holds, then  
	\begin{equation}
	\label{lemma_eq_8}
	\left(1+\alpha^2\right)\left( \frac{\varphi_2}{\varphi_4} \right)^2 + \alpha\left(  \frac{\varphi_2}{\varphi_4}  \right) + \alpha^2+\alpha+1=0.
	\end{equation}
	If $\alpha=1$, it is easy to obtain that $\beta=1$ from the definition of $\Gamma$. Moreover, $\varphi_2=0$ and thus $\tr_q\left(\frac{\varphi_2}{\varphi_4}\right)=0.$ In the following, we assume that $\alpha\neq 1$. Then after multiplying Eq. (\ref{lemma_eq_8}) by $\frac{\alpha^2+1}{\alpha^2}$  and simplifying, we get
	$$ \left( \frac{\alpha^2+1}{\alpha} \cdot \frac{\varphi_2}{\varphi_4}  \right)^2+ \frac{\alpha^2+1}{\alpha} \cdot \frac{\varphi_2}{\varphi_4} = \left( \frac{\alpha^2+1}{\alpha} \right)^2+\frac{\alpha^2+1}{\alpha}$$
	and thus 
	$$\frac{\varphi_2}{\varphi_4}=1 ~~\text{or}~~ \frac{\alpha^2+\alpha+1}{\alpha^2+1}.  $$
	It is clear that $\varphi_2+\varphi_4\neq0$. Suppose we have  $\frac{\varphi_2}{\varphi_4}=\frac{\alpha^2+\alpha+1}{\alpha^2+1}$. Moreover, $\frac{\varphi_2+\varphi_4}{\varphi_4}=\frac{\alpha}{\alpha^2+1}$ and $\frac{\varphi_1+\varphi_4}{\varphi_4}=\left( \frac{\varphi_2+\varphi_4}{\varphi_4} \right)^{2^{i}}=\frac{\alpha^{2^i}}{\alpha^{2^{i+1}}+1}.$ Furthermore, $\frac{\varphi_1+\varphi_4}{\varphi_2+\varphi_4}=\alpha^{2^i-1}$ and thus $\frac{\varphi_2+\varphi_4}{\varphi_4}=\left( \frac{\varphi_1+\varphi_4}{\varphi_2+\varphi_4} \right)^{\frac{1}{2^i-1}} = \alpha =\frac{\alpha}{\alpha^2+1},$ which is impossible. 
	
	Therefore, Eq. (\ref{lemma_eq_7}) does not hold and thus 
	\begin{equation}
	\label{2/4}
	\left(\frac{\varphi_2}{\varphi_4}\right)^2 = \frac{\varphi_2+\varphi_4}{\varphi_4}\alpha+\left( \frac{\varphi_2+\varphi_4}{\varphi_4}\alpha\right)^2.
	\end{equation}
	Clearly, $$\tr_q\left(\frac{\varphi_2}{\varphi_4}\right)=0.$$

\end{proof}

\section{The permutation property of the butterfly structure}

\label{permutation}

In this section, we firstly give a general necessary and sufficient condition about the permutation property of the function $V_i$ from the closed butterfly. Throughout what follows, we always assume $n$ is an odd integer.

\medskip

Recall that the univariate representation of $V_i$ have the following form
\begin{equation}
\label{f(x)}
f(x)=\epsilon_1x^{q\cdot(2^i+1)}+ \epsilon_2 x^{q\cdot 2^i+1} +  \epsilon_3x^{2^i+q} +  \epsilon_4 x^{2^i+1}, \quad \epsilon_j\in \mathbb{F}_q.
\end{equation}
Below we first present a necessary and sufficient conditions for $f(x)$ to be a permutation of $\mathbb{F}_{q^2}$ without 
imposing any additional restrictions on $\epsilon_j$. For simplicity of presentation, we denote
\begin{equation}	
\label{varphi_permutation}
\left\{
\begin{array}{lr}
\varphi_1=\epsilon_1\epsilon_3+\epsilon_2\epsilon_4  \\ 
\varphi_2=\epsilon_1\epsilon_2+\epsilon_3\epsilon_4 \\
\varphi_3=\epsilon_1^2+\epsilon_4^2 \\
 \varphi_4=\epsilon_1^2+\epsilon_2^2+\epsilon_3^2+\epsilon_4^2.
\end{array}
\right.
\end{equation}

The following proposition investigates the permutation property of $f(x)$ defined by (\ref{f(x)}) over $\gf_{q^2}$. 
\begin{Prop}
	\label{permutation_Vi}
	Let $q=2^n$, $f(x)$ be defined by (\ref{f(x)}), $h(x)=\epsilon_1x^{2^i+1}+\epsilon_2x^{2^i}+\epsilon_3x+\epsilon_4$ and $g(x)=x^{2^i+1}h(x)^{q-1}$. Define $\mu_{q+1} = \{ x\in\gf_{q^2} : x^{q+1}=1 \}$ and $$T = \left\{ \left( \frac{xy+1}{x+y}, \frac{xy}{x^2+y^2} \right) ~~\bigg|~~ x,y\in\mu_{q+1}\backslash\{1\}, y\neq x, x^q \right\}\subset \gf_{q}^2.$$ Then $f(x)$ permutes $\gf_{q^2}$ if and only if
	\begin{enumerate}[(1)]
		\item $\gcd\left( 2^i+1, q-1 \right)=1$;
		\item $h(x)=0$ has no solution in $\mu_{q+1}$;
		\item $g(x)=1$ if and only if $x=1$;
		\item   there does not exist some $(X,Y)\in T$ such that the following equation holds: 
		\begin{equation}
		\varphi_1X^{2^i}+\varphi_2 X + \varphi_3 + \varphi_4 \left(\sum_{j=0}^{i-1} Y^{2^j}  \right) =0,
		\end{equation}
		where $\varphi_j$ for $j=1,2,3,4$ are defined by (\ref{varphi_permutation}).
	\end{enumerate}
\end{Prop}

\begin{proof}
  It is clear that $f(x)=x^{2^i+1}h\left( x^{q-1} \right)$. According to Lemma \ref{lem1}, $f(x)$ permutes $\gf_{q^2}$ if and only if $\gcd\left( 2^i+1, q-1 \right)=1$ and $$g(x)=x^{2^i+1}h(x)^{q-1}=\frac{\epsilon_4x^{2^i+1}+\epsilon_3x^{2^i}+\epsilon_2x+\epsilon_1}{\epsilon_1x^{2^i+1}+\epsilon_2x^{2^i}+\epsilon_3x+\epsilon_4}$$
	permutes $\mu_{q+1}$, which obviously implies that $h(x)=0$ has no solution in $\mu_{q+1}$ and $g(x)=1$ if and only if $x=1$.  In the following, we assume that the conditions (1),(2) and (3) hold. Therefore, $g(x)$ permutes $\mu_{q+1}$ if and only if $g(x)+g(y)=0$ has no solution for $x,y\in\mu_{q+1}\backslash\{1\}$ with $x\neq y$. In fact, if $g(x)+g(y)=0$ for some $y=x^q$, then we have $g(x)=g(y)=g(x^q)=g(x)^q=g(x)^{-1}$ and  thus  $g(x)=1$, which means that $x=1$. 
	Thus we can only consider the conditions such that $g(x)+g(y)=0$ has no solution for $x,y\in\mu_{q+1}\backslash\{1\}$ with $y\neq x,x^q$.
Next, we prove the necessity and sufficiency of the condition (4).

{\bfseries The sufficiency of (4).}    
Suppose $g(x)+g(y)=0$, i.e.,
$$
\frac{\epsilon_4x^{2^i+1}+\epsilon_3x^{2^i}+\epsilon_2x+\epsilon_1}{\epsilon_1x^{2^i+1}+\epsilon_2x^{2^i}+\epsilon_3x+\epsilon_4} = \frac{\epsilon_4y^{2^i+1}+\epsilon_3y^{2^i}+\epsilon_2y+\epsilon_1}{\epsilon_1y^{2^i+1}+\epsilon_2y^{2^i}+\epsilon_3y+\epsilon_4}.
$$
After a routine calculation, we obtain
	\begin{eqnarray*}
		\varphi_1(x+y)(xy+1)^{2^i} + \varphi_2 (x+y)^{2^i}(xy+1)  + \varphi_3(x+y)^{2^i+1}+ \varphi_4 \left(x^{2^i}y+xy^{2^i}\right)=0,
	\end{eqnarray*}
	where  $\varphi_j$ for $j=1,2,3,4$ are as defined in (\ref{varphi_permutation}). By the previous discussion, we now only need to consider the case that  $(x+y)(xy+1)\neq 0$.
    Therefor, the above equation is equivalent to 
  \begin{equation}\label{Eq-prop1}
  \varphi_1\left(\frac{xy+1}{x+y}\right)^{2^i} + \varphi_2 \left(\frac{xy+1}{x+y}\right)+ \varphi_3+ \varphi_4 \left(\frac{x^{2^i}y+xy^{2^i}}{(x+y)^{2^i+1}}\right)=0.
  \end{equation}
  Note that 
  $$
  \frac{x^{2^i}y+xy^{2^i}}{(x+y)^{2^i+1}} =  \frac{x^{2^i+1}+y^{2^i+1}}{(x+y)^{2^i+1}}  + 1 = \left(\frac{x}{x+y}\right)^{2^i+1} + \left(\frac{y}{x+y}\right)^{2^i+1} + 1.
  $$
  It follows from Lemma \ref{Dickson_property} (1) that the coefficient of $\varphi_4$ can be expressed in terms of Dickson polynomial as 
  $$
  \frac{x^{2^i}y+xy^{2^i}}{(x+y)^{2^i+1}} = D_{2^i+1}\left(1, \frac{xy}{(x+y)^2}\right) + 1.
  $$
  In addition, by Lemma \ref{Dickson_property} (2), (4) and Lemma \ref{Dickson_2i-1}, 
  \begin{equation}
  \begin{split}
  	  D_{2^i+1}(x, a) = D_{2^i}(x, a) + aD_{2^i-1}(x, a) 
  	  & = x^{2^i} +  \sum\limits_{j=0}^{i-1}a^{2^j}x^{2^i-2^{j+1}+1}. 
  \end{split}
  \end{equation} Denote $X=\frac{xy+1}{x+y}$ and $Y=\frac{xy}{(x+y)^2}$.  It is straightforward that
$g(x) = g(y)$ can be rewritten as
 	\begin{equation}
 	\label{eq_XY}
 	\varphi_1X^{2^i}+\varphi_2 X + \varphi_3 + \varphi_4 \left(\sum_{j=0}^{i-1} Y^{2^j}  \right) =0.
 	\end{equation}
%
%
%
	Thus, if there exist some $x,y\in\mu_{q+1}$ with $y\neq x,x^q$ such that $g(x)+g(y)=0$ holds, there must exist some $(X,Y)\in T$ such that Eq. (\ref{eq_XY}) holds. Thus if the condition (4) holds, $g(x)$ permutes $\mu_{q+1}$. 

{\bfseries The necessity of (4).}  
On the contrary, if the condition (4) does not hold, which means that there exist some $(X,Y)\in T$ such that Eq. (\ref{eq_XY}) holds, then there must exist some $x,y\in\mu_{q+1}\backslash\{1\}$ with $y\neq x,x^q$ such that $g(x)+g(y)=0$, which implies that $g(x)$ does not permute $\mu_{q+1}$.  

On combining the sufficiency and necessity, we have proved the desired conclusion.
\end{proof}



{\bfseries Proof the permutation part in Theorem \ref{main_theorem}.}

In the following, we will prove the permutation part in Theorem \ref{main_theorem} by verifying the conditions in Proposition \ref{permutation_Vi}.

First of all, if $\alpha=1$, it is easy to obtain that $\beta=1$ from the definition of $\Gamma$ and
 $$f_i(x) = \left\{
\begin{array}{lr}
x^{q\cdot (2^i+1)}, ~\text{when}~ i ~\text{is even}  \\ 
x^{2^i+q}, ~\text{when}~ i ~\text{is odd},
\end{array}
\right. $$ clearly permutes $\gf_{q^2}$. Thus in the following, we assume that $\alpha\neq1$.  It suffices to show the four items of Proposition \ref{permutation_Vi}.

(1) Since $n$ is odd and $\gcd(i,n)=1$, we have $\gcd(2^i+1,2^n-1)=1$ due to the fact $\gcd(2^i+1,2^n-1) \mid \gcd(2^{2i}-1,2^n-1) = 2^{\gcd(2i,n)}-1=1$. 

(2) Next we show that $h(x)=0$ has no solution in $\mu_{q+1}\backslash\{1\}$ ($h(1)=\varphi_4^{1/2}\neq0$ according to the definition). Suppose that there exists some $x_0\in\mu_{q+1}\backslash\{1\}$ satisfying
 \begin{equation}
 \label{h(x_0)=0}
 \epsilon_1x_0^{2^i+1}+\epsilon_2x_0^{2^i}+\epsilon_3x_0+\epsilon_4 =0. 
 \end{equation}
 Raising Eq. (\ref{h(x_0)=0}) to the $q$-th power and re-arranging it according to $x_0^q=x_0^{-1}$, we obtain 
  \begin{equation}
 \label{h(x_0)=0_1}
 \epsilon_4x_0^{2^i+1}+\epsilon_3x_0^{2^i}+\epsilon_2x_0+\epsilon_1 =0. 
 \end{equation}
Summing $\epsilon_4 \times (\ref{h(x_0)=0})$ and $\epsilon_1 \times (\ref{h(x_0)=0_1})$ gives
 \begin{equation}
 \label{h(x_0)=0_2}
 \varphi_1x_0^{2^i}+\varphi_2x_0+\varphi_3=0.
 \end{equation}
Computing $\varphi_3 \times (\ref{h(x_0)=0_2}) + \varphi_1 \times (\ref{h(x_0)=0_2})^q\times x_0^{2^i} $ yields
 \begin{equation}
 \label{h(x_0)=0_3}
 \varphi_1\varphi_2x_0^{2^i-1}+\varphi_2\varphi_3x_0+\varphi_1^2+\varphi_3^2=0.
 \end{equation}
 Furthermore, by computing $(\ref{h(x_0)=0_3})\times x_0 + (\ref{h(x_0)=0_2}) \times \varphi_2$, we obtain 
 \begin{equation}
 \label{h(x_0)=0_4}
 \varphi_2\varphi_3 x_0^2+\left( \varphi_1^2+\varphi_2^2+\varphi_3^2 \right) x_0 + \varphi_2\varphi_3=0.
 \end{equation}
Note that in the above equation $\varphi_2\varphi_3\neq0$. Otherwise, we have $\varphi_1^2+\varphi_2^2=\varphi_3^2$. Recall that $\varphi_1^2+\varphi_2^2=\varphi_3(\varphi_3+\varphi_4)$
from \eqref{Eq-vaphi1+2} and \eqref{Eq-varphi3+4}.
Thus we obtain $\varphi_3\varphi_4=0$, which is in contradiction with $\varphi_4\neq 0$ in definition of $\Gamma$ and $\varphi_3\neq 0$ in Lemma \ref{property of Gamma} (1).
 which is also a contradiction. Thus Eq. (\ref{h(x_0)=0_4}) becomes 
 \begin{equation}
 \label{h(x_0)=0_5}
 x_0^2+\frac{\varphi_1^2+\varphi_2^2+\varphi_3^2}{\varphi_2\varphi_3}  x_0 + 1 =0.
 \end{equation}
Note that
 \begin{eqnarray*}
 \tr_q\left(\frac{\varphi_2\varphi_3}{\varphi_1^2+\varphi_2^2+\varphi_3^2}\right)=\tr_q\left( \frac{\varphi_2\varphi_3}{\varphi_3 \varphi_4} \right)
=\tr_q\left( \frac{\varphi_2}{\varphi_4} \right) =0.
 \end{eqnarray*}
This implies that Eq. (\ref{h(x_0)=0_5}) has a solution $x_0\in\gf_{q}$, which contradicts $\mu_{q+1}\backslash\{1\}$. Therefore,  $h(x)=0$ has no solution in $\mu_{q+1}$.
 
 (3) If there exists some $x_0\in\mu_{q+1}\backslash\{1\}$ such that $g(x_0)=1$, then we have
 \begin{equation}
 \label{g(x_0)=1}
 \left(\epsilon_1+\epsilon_4\right)x_0^{2^i+1}+\left(\epsilon_2+\epsilon_3\right)x_0^{2^i}+\left(\epsilon_2+\epsilon_3\right)x_0+\epsilon_1+\epsilon_4=0.
 \end{equation}
 
 According to Lemma \ref{mu_Fq}, we know that for any $x_0\in\mu_{q+1}\backslash\{1\}$, there exists a unique element $y_0\in\gf_{q}$ such that $x_0=\frac{y_0+\gamma}{y_0+\gamma^2}$, where $\gamma\in\gf_{2^2}\backslash\gf_2$. 
By plugging $x_0=\frac{y_0+\gamma}{y_0+\gamma^2}$ into Eq. (\ref{g(x_0)=1}) and a routine rearrangement, we obtain
 \begin{equation}\label{g(x_0)=2}
 y_0^{2^i}+y_0+\frac{\varepsilon_1+\varepsilon_4}{\epsilon_1+\epsilon_2+\epsilon_3+\epsilon_4}=0,
 \end{equation}
 where $\varepsilon_1, \varepsilon_4$ are defined as in \eqref{ieven} satisfying that $\varepsilon_1+\varepsilon_4=\epsilon_1+\epsilon_4$ for even $i$ and $\varepsilon_1+\varepsilon_4=\epsilon_2+\epsilon_3$ for odd $i$. 
In other words, $\varepsilon_1+\varepsilon_4$ corresponds to $\varphi_3+\varphi_4$ for even $i$ and $\varphi_3$ for odd $i$.
By Lemma \ref{property of Gamma} (2) and (3), we have 
 $$
 \tr_q\left( \frac{\varepsilon_1+\varepsilon_4}{\epsilon_1+\epsilon_2+\epsilon_3+\epsilon_4} \right) = 1.
 $$
This implies \eqref{g(x_0)=2} has no solution in $\mathbb{F}_q$.
Hence $g(x)=1$ if and only if $x=1$.
 
 (4)  Recall that $Y=\frac{xy}{x^2+y^2}$ for some $x,y\in\mu_{q+1}\backslash\{1\}$ with $x\neq y$ and thus
 $$\tr_q\left(Y\right)=\tr_q\left( \frac{y}{x+y} + \left( \frac{y}{x+y} \right)^2 \right) = 1,$$
 since $\frac{y}{x+y}\in\gf_{q^2}\backslash\gf_{q}$.  It is clear that Eq. (\ref{eq_XY}) required in Proposition \ref{permutation_Vi} is equivalent to 
\begin{eqnarray*}
\sum_{j=0}^{i-1} Y^{2^j}  &=&\frac{\varphi_1}{\varphi_4}X^{2^i}+\frac{\varphi_2}{\varphi_4}X+\frac{\varphi_3}{\varphi_4}\\
&=& \left(\frac{\varphi_2}{\varphi_4}X\right)^{2^i}+\frac{\varphi_2}{\varphi_4}X+\frac{\varphi_3}{\varphi_4}.
\end{eqnarray*}
By $\tr_q(Y)=1$ we have
\begin{eqnarray}
\tr_q\left(  \sum_{j=0}^{i-1} Y^{2^j}  \right) = \left\{
\begin{array}{lr}
0, ~\text{when}~ i ~\text{is even}  \\ 
1, ~\text{when}~ i ~\text{is odd},
\end{array}
\right.
\end{eqnarray}
on the other hand,  the expression on the right hand side satisfies $$\tr_q\left( \left(\frac{\varphi_2}{\varphi_4}X\right)^{2^i}+\frac{\varphi_2}{\varphi_4}X+\frac{\varphi_3}{\varphi_4} \right) = \left\{
 \begin{array}{lr}
 1, ~\text{when}~ i ~\text{is even}  \\ 
 0, ~\text{when}~ i ~\text{is odd},
 \end{array}
 \right. $$
according to Lemma \ref{property of Gamma}. It is clear that Eq. (\ref{eq_XY})  does not hold for any $X,Y\in\gf_q$.

Up to now, all the four items in Proposition  \ref{permutation_Vi} are confirmed. Hence the function $V_i(x,y)$ in Theorem \ref{main_theorem} permutes $\gf_{q}^2$.

\section{The boomerang uniformity of $V_i$ in Theorem \ref{main_theorem}}

\label{boomerang}

In this section, we will prove that the function 
$$V_i := (R_i(x,y), R_i(y,x))$$ 
with $R_i(x,y)=(x+\alpha y)^{2^i+1}+\beta y^{2^i+1}$  has boomerang uniformity $4$ when the pair $(\alpha, \beta)$ is taken from the set $\Gamma$ as in given in Theorem \ref{main_theorem}.  Here and hereafter, we assume that $n$ is odd, $q=2^n$ and $(\alpha,\beta)\in\Gamma$.  

First of all, the condition $\beta\neq (\alpha+1)^{2^i+1}$ in Lemma \ref{differential_uniformity_butterfly} corresponds to the condition $\varphi_4\neq 0$ in $\Gamma$. Hence the differential uniformity of with $R_i(x,y)=(x+\alpha y)^{2^i+1}+\beta y^{2^i+1}$ is at most  $4$ for any $(\alpha,\beta)\in\Gamma$. Furthermore, Canteaut et al. \cite{canteaut2018if} showed that if $V_i$ is APN then it operates on 6 bits. Therefore, the differential uniformity of $V_i$ is equal to $4$.  Since $V_i$ in Theorem \ref{main_theorem} permutes $\gf_{q}^2$ and hasdifferential uniformity $4$, we can use Lemma \ref{quadratic_boomerang} to show the boomerang uniformity of $V_i$.
For any $(a,b)\in\gf_{q}^2$, denote $$S_{V_i,(a,b)}(x,y) = V_i(x+a,y+b)+V_i(x,y)+V_i(a,b)$$ and $$\mathrm{Im}_{V_i,(a,b)} = \left\{ S_{V_i,(a,b)}(x,y) ~~|~~ (x,y)\in\gf_{q}^2  \right\}.$$ 
According to Lemma \ref{quadratic_boomerang}, we need to determine  $(a_1,b_1), (a_2,b_2)\in\gf_{q}^2\backslash\{(0,0)\}$ satisfying $S_{V_i,(a_1,b_1)}(a_2,b_2)=(0,0)$, and then to prove that for any such pairs  the equation $\mathrm{Im}_{V_i,(a_1,b_1)}=\mathrm{Im}_{V_i,(a_2,b_2)}$ holds.

\subsection{The solutions of $S_{V_i,(a_1,b_1)}(a_2,b_2)=(0,0)$}

The solution of the equation $S_{V_i,(a_1,b_1)}(a_2,b_2)=(0,0)$ is studied in the following proposition.

\begin{Prop}
	\label{relations}
	Let $V_i$ be defined as in Theorem \ref{main_theorem} with $(\alpha,\beta) \in \Gamma $ and $\varphi_j$ for $j=1,2,3,4$ defined as in (\ref{varphi}).
	Then the elements $(a_1,b_1), (a_2,b_2) \in\gf_{q}^2\backslash\{ (0,0) \}$such that 
	$$V_i(a_1+a_2, b_1+b_2)+V_i(a_1,b_1)+V_i(a_2,b_2)=(0,0)$$ are given as follows:
	\begin{enumerate}[(1)]
		\item $a_2=a_1$ and $b_2=b_1$;
		\item $a_2=\left( \frac{\varphi_2+\varphi_4}{\varphi_4}\alpha+1\right) a_1+\frac{\varphi_2+\varphi_4}{\varphi_4}b_1$ and  $b_2= \frac{\varphi_2+\varphi_4}{\varphi_4} a_1+  \frac{\varphi_2+\varphi_4}{\varphi_4}  \alpha b_1$;
		\item $a_2=\frac{\varphi_2+\varphi_4}{\varphi_4}\alpha a_1+\frac{\varphi_2+\varphi_4}{\varphi_4}b_1$ and  $b_2= \frac{\varphi_2+\varphi_4}{\varphi_4} a_1+\left( \frac{\varphi_2+\varphi_4}{\varphi_4}\alpha+1\right)b_1$.
	\end{enumerate}
\end{Prop}

\begin{proof}
Note that the equation $$S_{V_i,(a_1,b_1)}(a_2,b_2)=V_i(a_1+a_2, b_1+b_2)+V_i(a_1,b_1)+V_i(a_2,b_2)=(0,0)$$  can be rewritten as
	\begin{subequations} 
		\renewcommand\theequation{\theparentequation.\arabic{equation}}     
		\label{B_V=0}
		\small
		\begin{empheq}[left={\empheqlbrace\,}]{align}
		& (a_1+\alpha b_1 ) a_2^{2^i} + ( a_1^{2^i}+\alpha^{2^i}b_1^{2^i} ) a_2 + ( \alpha^{2^i}a_1+(\alpha^{2^i+1}+\beta)b_1 ) b_2^{2^i} + ( \alpha a_1^{2^i}+ (\alpha^{2^i+1}+\beta) b_1^{2^i}  ) b_2 =0 \label{eq_a1a2b1b2_1} \\ 
		&  ( (\alpha^{2^i+1}+\beta)a_1+\alpha^{2^i}b_1 )a_2^{2^i}+ ( (\alpha^{2^i+1}+\beta) a_1^{2^i}+ \alpha b_1^{2^i}  )a_2+  (\alpha a_1+ b_1 )b_2^{2^i} +   (\alpha^{2^i} a_1^{2^i} + b_1^{2^i} ) b_2=0. \label{eq_a1a2b1b2_2}
		\end{empheq}
	\end{subequations}

Let $\varphi_j$ for $j=1,2,3,4$ be defined by (\ref{varphi}).  Eliminating the terms $a_2^{2^i}$ in the above equations by computing $(\ref{eq_a1a2b1b2_1})\times \left( \left(\alpha^{2^i+1}+\beta\right)a_1+\alpha^{2^i}b_1  \right) + (\ref{eq_a1a2b1b2_2}) \times \left(a_1+\alpha b_1 \right)$, we obtain 
\begin{equation}
\label{a_2}
\lambda_1 a_2+\lambda_2 b_2^{2^i}+\lambda_3 b_2=0,
\end{equation}
where the coefficients are given by
\begin{equation*}	
\left\{
\begin{array}{lr}
\lambda_1 = \left( \varphi_1+\varphi_4 \right) a_1^{2^i}b_1+\left( \varphi_2+\varphi_4 \right) a_1b_1^{2^i} + (\varphi_{3}+\varphi_4)b_1^{2^i+1}  \\ 
\lambda_2 = \left(  \varphi_2+\varphi_4 \right) a_1^2 + \varphi_4a_1b_1 + \left( \varphi_2+\varphi_4 \right) b_1^2 \\
\lambda_3 = \left( \varphi_1+\varphi_4\right) a_1^{2^i+1} + \varphi_{3} a_1b_1^{2^i} +  \left( \varphi_2+\varphi_4 \right) b_1^{2^i+1},
\end{array}
\right.
\end{equation*}  where $\varphi_1, \varphi_2, \varphi_4$ are as defined in \eqref{varphi} and $\varphi_{3}$ is indeed $\varphi_{3,e}=(\alpha^{2^i+1}+\beta + 1)^2$ for even $i$. Here and hereafter, we use 
$\varphi_3$ to denote $\varphi_{3,e}$ for simplicity of notation.


\smallskip

When $b_1=0$, we have $a_1\neq0$, $\lambda_1=0, \lambda_2 = \left(  \varphi_2+\varphi_4 \right) a_1^2 $ and $ \lambda_3 = \left( \varphi_1+\varphi_4\right) a_1^{2^i+1}$. Moreover, Eq. (\ref{a_2}) becomes $\lambda_2 b_2^{2^i} = \lambda_3 b_2$. This together with Lemma \ref{property of Gamma} (1) implies 
$$b_2=0 \text{ or } b_2=\left( \frac{\varphi_1+\varphi_4}{\varphi_2+\varphi_4} \right)^{\frac{1}{2^i-1}}a_1 = \frac{\varphi_2+\varphi_4}{\varphi_4} a_1.   $$
Note that in the case of $b_1=0$, Eq. (\ref{eq_a1a2b1b2_1}) becomes $$ \left(\frac{a_2}{a_1}\right)^{2^i}+\frac{a_2}{a_1} = \left( \frac{\alpha b_2}{a_1} \right)^{2^i} + \frac{\alpha b_2}{a_1}. $$
Therefore, if $b_2=0$, then $a_2=a_1$; if $b_2=\frac{\varphi_2+\varphi_4}{\varphi_4} a_1$, then $a_2=\frac{\varphi_2+\varphi_4}{\varphi_4}\alpha a_1$ or $a_2=\frac{\varphi_2+\varphi_4}{\varphi_4}\alpha a_1+a_1$.

When $b_1\neq0$. Eliminating the terms $b_2^{2^i}$ by computing $(\ref{eq_a1a2b1b2_1})\times \left( (\alpha^{2^i+1}+\beta) a_1^{2^i}+ \alpha b_1^{2^i} \right) + (\ref{eq_a1a2b1b2_2}) \times \left( a_1^{2^i}+\alpha^{2^i}b_1^{2^i} \right)$, we obtain
\begin{equation}
\label{a_2^4}
\eta_1 a_2^{2^i} + \eta_2 b_2^{2^i}+\eta_3 b_2=0,
\end{equation}
where 
\begin{equation*}	
\left\{
\begin{array}{lr}
\eta_1 = \lambda_1  \\ 
\eta_2 = \left( \varphi_2+\varphi_4 \right) a_1^{2^i+1} + \varphi_3a_1^{2^i}b_1 + \left( \varphi_1+\varphi_4 \right) b_1^{2^i+1}\\
\eta_3 = \left(  \varphi_1+\varphi_4 \right) a_1^{2^{i+1}} + \varphi_4 a_1^{2^i}b_1^{2^i} +  \left(\varphi_1+\varphi_4 \right) b_1^{2^{i+1}}. 
\end{array}
\right.
\end{equation*} 


Furthermore, computing $(\ref{a_2})^{2^i}+\lambda_1^{2^i-1}\times(\ref{a_2^4})$, we eliminate the terms $a_2^{2^i}$ and obtain
\begin{equation}
\label{b2}
\lambda_2^{2^i}b_2^{2^{2i}-1}+\left( \lambda_1^{2^i-1}\eta_2 + \lambda_3^{2^i} \right) b_2^{2^i-1} + \lambda_1^{2^i-1} \eta_3 =0.
\end{equation}
Here we note that $\lambda_2\neq0$. Otherwise one has  $\left(  \varphi_2+\varphi_4 \right) a_1^2 + \varphi_4a_1b_1 + \left( \varphi_2+\varphi_4 \right) b_1^2=0,$
i.e.,
$$\left(\frac{\varphi_2+\varphi_4}{\varphi_4}\cdot\frac{a_1}{b_1}\right)^2+\frac{\varphi_2+\varphi_4}{\varphi_4}\cdot\frac{a_1}{b_1}+\left(\frac{\varphi_2+\varphi_4}{\varphi_4} \right)^2=0,$$
which is in contradiction with the fact $\tr_q\left(\frac{\varphi_2}{\varphi_4}\right)=0$ by Lemma \ref{property of Gamma} (4). 

In addition, since the differential uniformity of $V_i$ is  $4$, Eq. \eqref{b2} has three nonzero solutions $b_2=b_1, \bar{b} $ and $\bar{b}+b_1$ and we only need to obtain the expression of $\bar{b}$. Clearly, $\tilde{b}_2=b_1^{2^i-1}$ is a solution of 
\begin{equation}
\label{tilde_b2}
\lambda_2^{2^i}\tilde{b}_2^{2^i+1}+\left( \lambda_1^{2^i-1}\eta_2 + \lambda_3^{2^i} \right) \tilde{b}_2 + \lambda_1^{2^i-1} \eta_3 =0.
\end{equation}
Hence, Eq. (\ref{tilde_b2}) can be written as
$$\lambda_2^{2^i}\left(\tilde{b}_2 + b_1^{2^i-1}\right)\left( \tilde{b}_2^{2^i}+ b_1^{2^i-1} \tilde{b}_2^{2^i-1} + b_1^{2\cdot\left(2^i-1\right)} \tilde{b}_2^{2^i-2}+ \cdots + b_1^{\left(2^i-1\right)\cdot\left(2^i-1\right)} \tilde{b}_2 +c \right)=0,$$
where $c=\frac{\lambda_1^{2^i-1}\eta_3}{\lambda_2^{2^i}b_1^{2^i-1}}$. 
Now we consider the equation
\begin{equation}
\label{tilde_b2_2}
\tilde{b}_2^{2^i}+ b_1^{2^i-1} \tilde{b}_2^{2^i-1} + b_1^{2\cdot\left(2^i-1\right)} \tilde{b}_2^{2^i-2}+ \cdots + b_1^{\left(2^i-1\right)\cdot\left(2^i-1\right)} \tilde{b}_2 +c =0.
\end{equation}
Let $\hat{b}_2=\frac{1}{\tilde{b}_2+b_1^{2^i-1}}$. Then Eq. (\ref{tilde_b2_2}) becomes
$$\hat{b}_2^{2^i}+\frac{b_1^{2^i-1}}{c}\hat{b}_2+\frac{1}{c}=0,$$
i.e.,
\begin{equation}
\label{tilde_b2_3}
\left( \frac{c^{\frac{1}{2^i-1}}}{b_1} \hat{b}_2 \right)^{2^i} +  \frac{c^{\frac{1}{2^i-1}}}{b_1} \hat{b}_2 + \frac{c^{\frac{1}{2^i-1}}}{b_1^{2^i}} =0. 
\end{equation}
In addition, we have 
\begin{eqnarray*}
	c^{\frac{1}{2^i-1}} &=& \left(\frac{\lambda_1^{2^i-1}\eta_3}{\lambda_2^{2^i}b_1^{2^i-1}}  \right)^{\frac{1}{2^i-1}} \\
	&=& \frac{\lambda_1}{b_1}\left( \frac{ \left( \varphi_1+\varphi_4 \right) \left(   a_1^{2^{i+1}} + \frac{\varphi_4}{\varphi_1+\varphi_4} a_1^{2^i}b_1^{2^i} + b_1^{2^{i+1}}    \right)   }{\left( \varphi_2+\varphi_4 \right)^{2^i} \left(   a_1^{2^{i+1}} + \frac{\varphi_4^{2^i}}{\left( \varphi_2+\varphi_4\right)^{2^i}} a_1^{2^i}b_1^{2^i} + b_1^{2^{i+1}}    \right) }  \right)^{\frac{1}{{2^i}-1}} \\
	&=& \frac{\lambda_1}{b_1}\left( \frac{ \varphi_1+\varphi_4}{\left( \varphi_2+\varphi_4\right)^{2^i}} \right)^{\frac{1}{{2^i}-1}} \\
	&=& \frac{\lambda_1}{b_1\varphi_4},
\end{eqnarray*}
where the last two equalities follow from Lemma \ref{property of Gamma} (1).
Moreover, 
\begin{eqnarray*}
	\frac{c^{\frac{1}{{2^i}-1}}}{b_1^{2^i}} &=& \frac{\lambda_1}{\varphi_4b_1^{2^i+1}} \\
	&=& \frac{\left( \varphi_1+\varphi_4 \right) a_1^{2^i}+\left( \varphi_2+\varphi_4 \right) a_1b_1^{2^i-1} +\left(  \varphi_3+\varphi_4   \right) b_1^{2^i}}{\varphi_4b_1^{2^i}} \\
	&=& \left( \frac{\left( \varphi_2+\varphi_4\right)a_1}{\varphi_4b_1} \right)^{2^i} + \frac{\left( \varphi_2+\varphi_4\right)a_1}{\varphi_4b_1} + \frac{\varphi_3+\varphi_4}{\varphi_4} \\
	&=& \left( \frac{\left( \varphi_2+\varphi_4\right)a_1}{\varphi_4b_1} + u \right)^{2^i} + \frac{\left( \varphi_2+\varphi_4\right)a_1}{\varphi_4b_1} +u, 
\end{eqnarray*}
where $u=\frac{\varphi_2+\varphi_4}{\varphi_4}\alpha$ due to the second item of Lemma \ref{property of Gamma}. Hence, from Eq. (\ref{tilde_b2_3}), we have 
$$\frac{c^{\frac{1}{{2^i}-1}}}{b_1} \hat{b}_2 \in \left\{\frac{\left( \varphi_2+\varphi_4\right)a_1}{\varphi_4b_1} + u, \frac{\left( \varphi_2+\varphi_4\right)a_1}{\varphi_4b_1} + u+1\right\},  $$
which means that there are exactly two solutions in $\gf_{q}$ for Eq. (\ref{tilde_b2_2}). W.L.O.G., we only consider the first expression here. Namely, we get 
\begin{eqnarray*}
	\hat{b}_2 &=& \frac{b_1}{c^{\frac{1}{{2^i}-1}}}\left( \frac{\left( \varphi_2+\varphi_4\right)a_1}{\varphi_4b_1} + u\right) \\
	&=& \frac{\left( \varphi_2+\varphi_4\right)a_1b_1+\varphi_4ub_1^2}{\lambda_1}.
\end{eqnarray*}
Thus,
$$\tilde{b}_2=\frac{1}{\hat{b}_2}+b_1^{2^i-1} = \frac{\lambda_1}{\left( \varphi_2+\varphi_4\right)a_1b_1+\varphi_4ub_1^2}+b_1^{2^i-1}$$
is one solution of Eq. (\ref{tilde_b2_2}). Furthermore, one solution of Eq. (\ref{b2}) is 
\begin{eqnarray*}
	b_2 &=& \left(\tilde{b}_2\right)^{\frac{1}{{2^i-1}}} \\
	&=& \left( \frac{\left( \varphi_1+\varphi_4 \right)a_1^{2^i}+\varphi_4u^{2^i}b_1^{2^i}}{\left( \varphi_2+\varphi_4 \right)a_1+\varphi_4 u b_1} \right)^{\frac{1}{{2^i}-1}}\\
	&=& \left(\frac{\varphi_1+\varphi_4}{\varphi_2+\varphi_4}\right)^{\frac{1}{{2^i}-1}}\cdot \left(\frac{a_1^{2^i}+\frac{\varphi_4}{\varphi_1+\varphi_4}u^{2^i}b_1^{2^i}}{a_1+\frac{\varphi_4}{\varphi_2+\varphi_4}ub_1}\right)^{\frac{1}{{2^i}-1}} \\
	&=& \frac{\varphi_2+\varphi_4}{\varphi_4}\left(a_1+\frac{\varphi_4}{\varphi_2+\varphi_4}ub_1\right) (\text{by the first item of Lemma \ref{property of Gamma}})\\
	&=& \frac{\varphi_2+\varphi_4}{\varphi_4} a_1+  \frac{\varphi_2+\varphi_4}{\varphi_4}  \alpha b_1 (\text{recall  that } u=\frac{\varphi_2+\varphi_4}{\varphi_4}\alpha).
\end{eqnarray*}
It follows directly from Eq. (\ref{a_2}) that
\begin{eqnarray*}
	a_2 &=& \frac{\lambda_2}{\lambda_1} b_2^{2^i}+\frac{\lambda_3}{\lambda_1}b_2\\
	&=&\left( \frac{\varphi_2+\varphi_4}{\varphi_4}\alpha+1\right) a_1+\frac{\varphi_2+\varphi_4}{\varphi_4}b_1. 
\end{eqnarray*}

\end{proof}

\subsection{The proof of $\mathrm{Im}_{V_i,(a_1,b_1)}=\mathrm{Im}_{V_i,(a_2,b_2)}$}

In this subsection, we prove that for any  $(a_1,b_1), (a_2,b_2)\in\gf_{q}^2\backslash\{(0,0)\}$ satisfying $S_{V_i,(a_1,b_1)}(a_2,b_2)=(0,0)$, $\mathrm{Im}_{V_i,(a_1,b_1)}=\mathrm{Im}_{V_i,(a_2,b_2)}$. 

According to Eq. (\ref{B_V=0}), we know that for any $(a_1,b_1)\in\gf_{q}^2$, $S_{V_i,(a_1,b_1)}(x,y)$ can be represented as 
$$ S_{V_i,(a_1,b_1)}(x,y) = A_1 \begin{bmatrix}
x^{2^i} \\
x 
\end{bmatrix}+B_1\begin{bmatrix}
y^{2^i} \\
y 
\end{bmatrix}, $$
where 
$$A_1 = \begin{bmatrix}
a_1+\alpha b_1, & a_1^{2^i}+\alpha^{2^i}b_1^{2^i} \\
(\alpha^{2^i+1}+\beta)a_1+\alpha^{2^i}b_1, & (\alpha^{2^i+1}+\beta)a_1^{2^i}+\alpha b_1^{2^i} 
\end{bmatrix} \triangleq \begin{bmatrix}
a_{11}, & a_{12} \\
a_{13}, & a_{14} 
\end{bmatrix}$$ 
and
$$B_1 = \begin{bmatrix}
\alpha^{2^i}a_1+(\alpha^{2^i+1}+\beta)b_1, & \alpha a_1^{2^i}+(\alpha^{2^i+1}+\beta)b_1^{2^i} \\
\alpha a_1+ b_1, & \alpha^{2^i}a_1^{2^i}+b_1^{2^i}
\end{bmatrix} \triangleq \begin{bmatrix}
b_{11}, & b_{12} \\
b_{13}, & b_{14} 
\end{bmatrix}.$$

For the three relations between  $(a_1,b_1), (a_2,b_2)\in\gf_{q}^2\backslash\{ (0,0) \}$ presented in Proposition \ref{relations} such that $S_{V_i,(a_1,b_1)(a_2,b_2)}=(0,0)$, it is clear that if $a_2=a_1$ and $b_2=b_1$, we have $\mathrm{Im}_{V_i,(a_1,b_1)}=\mathrm{Im}_{V_i,(a_2,b_2)}$. In addition, if we have proved that $\mathrm{Im}_{V_i,(a_1,b_1)}=\mathrm{Im}_{V_i,(a_2,b_2)}$ holds for the second relation in Proposition \ref{relations}, then so does it for the third relation since the sum of two same subspace equals to the subspace. Therefore, it suffices to show that $\mathrm{Im}_{V_i,(a_1,b_1)}=\mathrm{Im}_{V_i,(a_2,b_2)}$ holds for the second relation in Proposition \ref{relations}. 
Below we will again use $\varphi_3$ to denote $\varphi_{3}$ for $\varphi_{3,e}$ for simplicity.

\medskip

Let $u=\frac{\varphi_2+\varphi_4}{\varphi_4}\alpha$. Then $u^{2^i}=u+\frac{\varphi_3+\varphi_4}{\varphi_4}$. Moreover, $a_2= (u+1) a_1+ \frac{\varphi_2+\varphi_4}{\varphi_4}  b_1$ and  $b_2= \frac{\varphi_2+\varphi_4}{\varphi_4} a_1+ub_1$. Furthermore, we get 
\begin{eqnarray*}
a_2^{2^i} &=& \left(u^{2^i}+1\right)a_1^{2^i}+\left(\frac{\varphi_2+\varphi_4}{\varphi_4}\right)^{2^i}b_1^{2^i} \\
&=& \left( u + \frac{\varphi_3}{\varphi_4} \right) a_1^{2^i} + \frac{\varphi_1+\varphi_4}{\varphi_4} b_1^{2^i}
\end{eqnarray*}
and 
\begin{eqnarray*}
b_2^{2^i} &=& \left(\frac{\varphi_2+\varphi_4}{\varphi_4}\right)^{2^i}a_1^{2^i} + u^{2^i} b_1^{2^i} \\
&=&\frac{\varphi_1+\varphi_4}{\varphi_4} a_1^{2^i} +  \left( u + \frac{\varphi_3+\varphi_4}{\varphi_4} \right)b_1^{2^i}.
\end{eqnarray*}
Therefore, in $S_{V_i,(a_2,b_2)}(x,y)$,  we have 
$$A_2 = \begin{bmatrix}
a_2+\alpha b_2, & a_2^{2^i}+\alpha^{2^i}b_2^{2^i} \\
(\alpha^{2^i+1}+\beta)a_2+\alpha^{2^i}b_2, & (\alpha^{2^i+1}+\beta)a_2^{2^i}+\alpha b_2^{2^i} 
\end{bmatrix} \triangleq \begin{bmatrix}
a_{21}, & a_{22} \\
a_{23}, & a_{24} 
\end{bmatrix},$$
and 
$$B_2 = \begin{bmatrix}
\alpha^{2^i}a_2+(\alpha^{2^i+1}+\beta)b_2, & \alpha a_2^{2^i}+(\alpha^{2^i+1}+\beta)b_2^{2^i} \\
\alpha a_2+ b_2, & \alpha^{2^i}a_2^{2^i}+b_2^{2^i}
\end{bmatrix} \triangleq \begin{bmatrix}
b_{21}, & b_{22} \\
b_{23}, & b_{24} 
\end{bmatrix},$$
where the explicit expressions of entries in $A_2$ and $B_2$ in terms of $a_1, b_1$ are given as follows:
$$
\begin{array}{rcl}
		a_{21} &=& a_2+\alpha b_2 \\
&=& (u+1) a_1+ \frac{\varphi_2+\varphi_4}{\varphi_4}  b_1 +  \frac{\varphi_2+\varphi_4}{\varphi_4} \alpha a_1+u \alpha b_1 \\
&=& \left(u+1+\frac{\varphi_2+\varphi_4}{\varphi_4} \alpha\right) a_1 + \left( \frac{\varphi_2+\varphi_4}{\varphi_4} + u \alpha \right)b_1 \\
&=&a_1 + \left(\alpha^2+1\right) \frac{\varphi_2+\varphi_4}{\varphi_4} b_1 (\text{recall that}~~u=\frac{\varphi_2+\varphi_4}{\varphi_4}\alpha),
\\  \\
	a_{22} &=& a_2^{2^i}+\alpha^{2^i}b_2^{2^i} \\
&=& a_1^{2^i} + \left(\alpha^{2^{i+1}}+1\right) \frac{\varphi_1+\varphi_4}{\varphi_4} b_1^{2^i} (\text{due to the first item of Lemma \ref{property of Gamma}}),
\\ \\
a_{23}	& =& (\alpha^{2^i+1}+\beta)a_2+\alpha^{2^i}b_2 \\
&=&\left( \left( \alpha^{2^{i}+1} +\beta\right) (u+1) + \alpha^{2^i}  \frac{\varphi_2+\varphi_4}{\varphi_4} \right) a_1 + \left( \left(\alpha^{2^{i}+1}+\beta\right)\frac{\varphi_2+\varphi_4}{\varphi_4} + \alpha^{2^i}u  \right)b_1\\
&=&\left(\frac{\left( \varphi_2+\varphi_4\right)\left(\varphi_1+\varphi_4\right)}{\varphi_4} + \alpha^{2^i+1}+\beta \right) a_1 + \frac{\varphi_2+\varphi_4}{\varphi_4} \beta b_1,
\\ \\
a_{24}	&=& (\alpha^{2^i+1}+\beta)a_2^{2^i}+\alpha b_2^{2^i} \\
&=&\left(  \left(\alpha^{2^i+1}+\beta\right) \left(u + \frac{\varphi_3}{\varphi_4}\right) + \alpha \frac{\varphi_1+\varphi_4}{\varphi_4} \right)a_1^{2^i} \\
& & + \left( \left(\alpha^{2^i+1}+\beta\right)\frac{\varphi_1+\varphi_4}{\varphi_4} + \alpha\left(u+\frac{\varphi_3+\varphi_4}{\varphi_4}\right)\right) b_1^{2^i} \\
&=&\left( \frac{\left( \varphi_2+\varphi_4\right)\left(\varphi_1+\varphi_4\right)}{\varphi_4} + \alpha^{2^i+1}+\beta \right) a_1^{2^i} + \frac{\varphi_1+\varphi_4}{\varphi_4} \beta b_1^{2^i} (\text{due to (\ref{varphi_1+2+3+4}) and  (\ref{lemma_eq_5})}),
\\ \\
 b_{21} &=& \alpha^{2^i}a_2+(\alpha^{2^i+1}+\beta)b_2 = \left( \alpha^{2^i}+ \frac{\varphi_2+\varphi_4}{\varphi_4} \beta  \right)a_1 + \frac{\left( \varphi_2+\varphi_4\right)\left(\varphi_1+\varphi_4\right)}{\varphi_4}  b_1, \\
b_{22} &=&  \alpha a_2^{2^i}+(\alpha^{2^i+1}+\beta)b_2^{2^i} =\left( \alpha+\frac{\varphi_1+\varphi_4}{\varphi_4}\beta\right) a_1^{2^i} + \frac{\left( \varphi_2+\varphi_4\right)\left(\varphi_1+\varphi_4\right)}{\varphi_4} b_1^{2^i}, \\
b_{23} &=& \alpha a_2+ b_2 = \left(\frac{\varphi_2+\varphi_4}{\varphi_4}(\alpha^2+1)+ \alpha \right) a_1, \\
b_{24} &=& \alpha^{2^i}a_2^{2^i}+b_2^{2^i} = \left(\frac{\varphi_1+\varphi_4}{\varphi_4}(\alpha^{2^{i+1}}+1)+ \alpha^{2^i} \right) a_1^{2^i}.  
\end{array}
$$

Note that the determinants of $A_1$ and $B_1$ are 
\begin{eqnarray*}
	\mathrm{Det}(A_1) &=& a_{11}a_{14}+a_{12}a_{13} \\
	&=& \left( \varphi_1+\varphi_4\right)a_1^{2^i}b_1 + \left(\varphi_2+\varphi_4\right) a_1b_1^{2^i}+\left(\varphi_3+\varphi_4\right)b_1^{2^i+1},
\end{eqnarray*}
and 
\begin{eqnarray*}
	\mathrm{Det}(B_1) &=& b_{11}b_{14}+b_{12}b_{13} \\
	&=&\left(\varphi_3+\varphi_4\right)a_1^{2^i+1} + \left( \varphi_2+\varphi_4\right)a_1^{2^i}b_1 + \left(\varphi_1+\varphi_4\right) a_1b_1^{2^i}.
\end{eqnarray*}
Now we consider the necessary and sufficient conditions such that $	\mathrm{Det}(A_1)=0$. Clearly, from  $	\mathrm{Det}(A_1)=0$, we have $b_1=0$ or 
$$\left( \varphi_1+\varphi_4\right)\left(\frac{a_1}{b_1}\right)^{2^i} + \left(\varphi_2+\varphi_4\right) \frac{a_1}{b_1}+\varphi_3+\varphi_4=0,$$
namely,
$$\left( \frac{\varphi_2+\varphi_4}{\varphi_4}\cdot \frac{a_1}{b_1} \right)^{2^i}+\frac{\varphi_2+\varphi_4}{\varphi_4}\cdot \frac{a_1}{b_1}=\frac{\varphi_3+\varphi_4}{\varphi_4}$$
and thus $a_1=\alpha b_1$ or $ \left( \alpha+\frac{\varphi_4}{\varphi_2+\varphi_4} \right)b_1$ due to Lemma \ref{property of Gamma}. 
Therefore,  $	\mathrm{Det}(A_1)=0 $ if and only if $b_1=0$ or $a_1=\alpha b_1$ or $ \left( \alpha+\frac{\varphi_4}{\varphi_2+\varphi_4} \right)b_1$.
Similarly, $	\mathrm{Det}(B_1)=0 $ if and only if $a_1=0$ or $b_1=\alpha a_1$ or $ \left( \alpha+\frac{\varphi_4}{\varphi_2+\varphi_4} \right)a_1$.

It is easy to verify that $\mathrm{Det}(A_1)=0 $ and  $	\mathrm{Det}(B_1)=0 $ holds at the same time if and only if 
\begin{enumerate}[(i)]
	\item $ \alpha = 1, a_1= b_1$;
	\item $\alpha+\frac{\varphi_4}{\varphi_2+\varphi_4} = 1, a_1=b_1$;
	\item $\alpha\left( \alpha+\frac{\varphi_4}{\varphi_2+\varphi_4}   \right) = 1,  a_1=\alpha b_1$.
\end{enumerate}
If $\alpha+\frac{\varphi_4}{\varphi_2+\varphi_4} = 1$, then $ \frac{\varphi_2}{\varphi_4} = \frac{\alpha}{\alpha+1}.$ Recall that (\ref{2/4}) holds, namely,
$$\left(\frac{\varphi_2}{\varphi_4}\right)^2 = \frac{\varphi_2+\varphi_4}{\varphi_4}\alpha+\left( \frac{\varphi_2+\varphi_4}{\varphi_4}\alpha\right)^2.$$
Plugging $ \frac{\varphi_2}{\varphi_4} = \frac{\alpha}{\alpha+1}$ into the above equation and simplifying, we obtain $\alpha=1$, implying $\frac{\varphi_4}{\varphi_2+\varphi_4}=0$, which is impossible. If $\alpha\left( \alpha+\frac{\varphi_4}{\varphi_2+\varphi_4}   \right) = 1,$ then $\frac{\varphi_2}{\varphi_4}=\frac{\alpha^2+\alpha+1}{\alpha^2+1}=\frac{1}{\alpha+1}+\frac{1}{\alpha^2+1}+1,$ which is also impossible since $\tr_q\left(\frac{\varphi_2}{\varphi_4}\right)=0.$  Therefore, $\mathrm{Det}(A_1)=0 $ and  $	\mathrm{Det}(B_1)=0 $ holds at the same time if and only if $ \alpha = 1, a_1= b_1$, under which it is clear that $\mathrm{Im}_{V_i,(a_1,b_1)}=\mathrm{Im}_{V_i,(a_2,b_2)}$.

\medskip

Next, we consider the following two cases: 
\begin{enumerate}[(i)]
	\item $\mathrm{Det}(B_1)\neq0$;
	\item $\mathrm{Det}(A_1)\neq0$.
\end{enumerate} 

It is clear that $\mathrm{Im}_{V_i,(a_1,b_1)}=\mathrm{Im}_{V_i,(a_2,b_2)}$ if there exist some invertible matrix $P$ such that $PA_1=A_2$ and $PB_1=B_2$.

  As for (i), it suffices to show that 
\begin{equation}
\label{matrix1}
B_2B_1^{-1}A_1=A_2.
\end{equation}

After computing, we know that (\ref{matrix1}) is
\begin{eqnarray*}
	&& \begin{bmatrix}
		b_{21}b_{14}a_{11}+b_{21}b_{12}a_{13}+b_{22}b_{13}a_{11}+b_{22}b_{11}a_{13}, & b_{21}b_{14}a_{12}+b_{21}b_{12}a_{14}+b_{22}b_{13}a_{12}+b_{22}b_{11}a_{14} \\
		b_{23}b_{14}a_{11}+b_{23}b_{12}a_{13}+b_{24}b_{13}a_{11}+b_{24}b_{11}a_{13}, & b_{23}b_{14}a_{12}+b_{23}b_{12}a_{14}+b_{24}b_{13}a_{12}+b_{24}b_{11}a_{14}
	\end{bmatrix} \\
	&& =  \mathrm{Det}(B_1)   \begin{bmatrix}
		a_{21}, & a_{22} \\
		a_{23}, & a_{24}
	\end{bmatrix}. 
\end{eqnarray*} 

After complicated computation and simplification, we get 
\begin{equation*}	
\left\{
\begin{array}{lr}
b_{14}a_{11}+b_{12}a_{13}=\left(\varphi_1+\varphi_4\right)a_1^{2^i+1}+\varphi_3 a_1b_1^{2^i}+\left( \varphi_2+\varphi_4 \right) b_1^{2^i+1} \\
b_{14}a_{12}+b_{12}a_{14}= \left(\varphi_1+\varphi_4\right)a_1^{2^{i+1}} + \varphi_4 a_1^{2^i}b_1^{2^i} + \left(\varphi_1+\varphi_4\right)b_1^{2^{i+1}}  \\
b_{13}a_{11}+b_{11}a_{13}= \left(\varphi_2+\varphi_4\right)a_1^{2} + \varphi_4 a_1b_1 + \left(\varphi_2+\varphi_4\right)b_1^{2} \\
b_{13}a_{12}+b_{11}a_{14}= \left(\varphi_2+\varphi_4\right)a_1^{2^i+1}+\varphi_3 a_1^{2^i}b_1+\left( \varphi_1+\varphi_4 \right) b_1^{2^i+1}.
\end{array}
\right.
\end{equation*} 
Moreover, we have
\begin{enumerate}
	\item \begin{eqnarray*}
		\small
		&& b_{21}b_{14}a_{11}+b_{21}b_{12}a_{13}+b_{22}b_{13}a_{11}+b_{22}b_{11}a_{13}\\
		&=& \left(\varphi_3+\varphi_4\right)a_1^{2^i+2} + \left( \varphi_3\alpha^{2^i} + \frac{\varphi_3\left(\varphi_2+\varphi_4\right)}{\varphi_4}\beta +\frac{\left(\varphi_2+\varphi_4\right)^2\left(\varphi_1+\varphi_4\right)}{\varphi_4}   \right) a_1^2b_1^{2^i} \\ 
		&+& \left( \left(\varphi_2+\varphi_4\right)\alpha^{2^i}+\frac{\left(\varphi_2+\varphi_4\right)^2}{\varphi_4}\beta+\frac{\left(\varphi_3+\varphi_4\right)\left(\varphi_2+\varphi_4\right)\left(\varphi_1+\varphi_4\right)}{\varphi_4}  \right) a_1b_1^{2^i+1}\\
		&+&\left( \varphi_4\alpha + \left( \varphi_1+\varphi_4 \right)\beta +  \frac{\left(\varphi_2+\varphi_4\right)\left(\varphi_1+\varphi_4\right)^2}{\varphi_4} \right)a_1^{2^i+1}b_1+\left( \left(\varphi_2+\varphi_4\right)\alpha+\frac{\left(\varphi_2+\varphi_4\right)\left(\varphi_1+\varphi_4\right)}{\varphi_4}\beta  \right)a_1^{2^i}b_1^2\\
		&=&\left(\varphi_3+\varphi_4\right)a_1^{2^i+2}+\left(\varphi_1+\varphi_4\right)a_1^2b_1^{2^i}+\frac{\left(\alpha^2+1\right)\left(\varphi_2+\varphi_4\right)\left(\varphi_1+\varphi_4\right)}{\varphi_4} a_1b_1^{2^i+1} \\
		&+&\frac{(\alpha^4+\beta^2+1)\left(\varphi_2+\varphi_4 \right)}{\varphi_4}a_1^{2^i+1}b_1 + \frac{\left(\alpha^2+1\right)\left(\varphi_2+\varphi_4\right)^2}{\varphi_4}a_1^{2^i}b_1^2,
	\end{eqnarray*}
	\item \begin{eqnarray*}
		\small
		& & b_{21}b_{14}a_{12}+b_{21}b_{12}a_{14}+b_{22}b_{13}a_{12}+b_{22}b_{11}a_{14} \\
		&=& \left(\varphi_3+\varphi_4\right)a_1^{2^{i+1}+1} + \left( \varphi_4\alpha^{2^i} + \left(\varphi_2+\varphi_4\right)\beta +\frac{\left(\varphi_2+\varphi_4\right)^2\left(\varphi_1+\varphi_4\right)}{\varphi_4} \right)a_1^{2^i+1}b_1^{2^i}\\
		&+& \left( (\varphi_1+\varphi_4)\alpha^{2^i} + \frac{\left(\varphi_2+\varphi_4\right)\left(\varphi_1+\varphi_4\right)}{\varphi_4}\beta \right)a_1b_1^{2^{i+1}} \\
		&+& 	\left( \frac{\left(\varphi_2+\varphi_4\right)\left(\varphi_1+\varphi_4\right)^2}{\varphi_4} + \varphi_3\alpha +\frac{\varphi_3(\varphi_1+\varphi_4)}{\varphi_4}\beta  \right)a_1^{2^{i+1}}b_1\\
		&+&\left( \left(\varphi_2+\varphi_4\right)\left(\varphi_1+\varphi_4\right)+\left(\varphi_1+\varphi_4\right)\alpha+\frac{(\varphi_1+\varphi_4)^2}{\varphi_4}\beta + \frac{\varphi_3\left(\varphi_2+\varphi_4\right)\left(\varphi_1+\varphi_4\right)}{\varphi_4}  \right) a_1^{2^i}b_1^{2^{i}+1} \\
		&=& \left(\varphi_3+\varphi_4\right)a_1^{2^{i+1}+1} + \frac{ \left(\alpha^{2^{i+2}}+\beta^2+1\right)\left(\varphi_1+\varphi_4\right)}{\varphi_4}a_1^{2^i+1}b_1^{2^i}+ \frac{\left(\alpha^{2^{i+1}}+1\right)\left(\varphi_1+\varphi_4\right)^2}{\varphi_4}a_1b_1^{2^{i+1}}\\
		&+&\left(\varphi_2+\varphi_4\right)a_1^{2^{i+1}}b_1+\frac{\left(\alpha^{2^{i+1}}+1\right)\left(\varphi_1+\varphi_4\right)\left(\varphi_2+\varphi_4\right)}{\varphi_4}a_1b_1^{2^{i+1}},
	\end{eqnarray*}
	\item 
	\begin{eqnarray*}
		& & b_{23}b_{14}a_{11}+b_{23}b_{12}a_{13}+b_{24}b_{13}a_{11}+b_{24}b_{11}a_{13}\\
		&=& \frac{\left(\varphi_1+\varphi_2\right)^2\beta}{\varphi_4}a_1^{2^i+2}+\frac{\varphi_3\left(\varphi_1+\varphi_4\right)\beta}{\varphi_4}a_1^2b_1^{2^i} +\frac{\left(\varphi_1+\varphi_4\right)\left(\varphi_2+\varphi_4\right)\beta}{\varphi_4}a_1b_1^{2^i+1} \\
		&+& \left(\varphi_2+\varphi_4\right) \beta a_1^{2^i+1}b_1 + \frac{\left(\varphi_2+\varphi_4\right)^2\beta}{\varphi_4} a_1^{2^i}b_1^2
	\end{eqnarray*}
	\item 
	\begin{eqnarray*}
		&&b_{23}b_{14}a_{12}+b_{23}b_{12}a_{14}+b_{24}b_{13}a_{12}+b_{24}b_{11}a_{14} \\
		&=& \frac{\left(\varphi_1+\varphi_2\right)^2\beta}{\varphi_4}a_1^{2^{i+1}+1}+\left(\varphi_1+\varphi_4\right)\beta a_1^{2^i+1}b_1^{2^i} +\frac{\left(\varphi_1+\varphi_4\right)^2\beta}{\varphi_4}a_1b_1^{2^{i+1}} \\
		&+& \frac{\varphi_3\left(\varphi_2+\varphi_4\right)\beta}{\varphi_4} a_1^{2^{i+1}}b_1 + \frac{\left(\varphi_1+\varphi_4\right)\left(\varphi_2+\varphi_4\right)\beta}{\varphi_4} a_1^{2^i}b_1^{2^i+1}.
	\end{eqnarray*}
\end{enumerate} 

Furthermore, after computing and simplifying, we have 
\begin{enumerate}
	\item \begin{eqnarray*}
	&& \mathrm{Det}(B_1) a_{21} \\
	&=& \left(\varphi_3+\varphi_4\right)a_1^{2^i+2}+\left(\varphi_1+\varphi_4\right)a_1^2b_1^{2^i}+\frac{\left(\alpha^2+1\right)\left(\varphi_2+\varphi_4\right)\left(\varphi_1+\varphi_4\right)}{\varphi_4} a_1b_1^{2^i+1} \\
	&+&\frac{(\alpha^4+\beta^2+1)\left(\varphi_2+\varphi_4 \right)}{\varphi_4}a_1^{2^i+1}b_1 + \frac{\left(\alpha^2+1\right)\left(\varphi_2+\varphi_4\right)^2}{\varphi_4}a_1^{2^i}b_1^2,
	\end{eqnarray*}
\item \begin{eqnarray*}
	&& \mathrm{Det}(B_1) a_{22} \\
	&=& \left(\varphi_3+\varphi_4\right)a_1^{2^{i+1}+1} + \frac{ \left(\alpha^{2^{i+2}}+\beta^2+1\right)\left(\varphi_1+\varphi_4\right)}{\varphi_4}a_1^{2^i+1}b_1^{2^i}+ \frac{\left(\alpha^{2^{i+1}}+1\right)\left(\varphi_1+\varphi_4\right)^2}{\varphi_4}a_1b_1^{2^{i+1}}\\
	&+&\left(\varphi_2+\varphi_4\right)a_1^{2^{i+1}}b_1+\frac{\left(\alpha^{2^{i+1}}+1\right)\left(\varphi_1+\varphi_4\right)\left(\varphi_2+\varphi_4\right)}{\varphi_4}a_1b_1^{2^{i+1}},
\end{eqnarray*}
\item 
\begin{eqnarray*}
	&& \mathrm{Det}(B_1) a_{23} \\
&=& \frac{\left(\varphi_1+\varphi_2\right)^2\beta}{\varphi_4}a_1^{2^i+2}+\frac{\varphi_3\left(\varphi_1+\varphi_4\right)\beta}{\varphi_4}a_1^2b_1^{2^i} +\frac{\left(\varphi_1+\varphi_4\right)\left(\varphi_2+\varphi_4\right)\beta}{\varphi_4}a_1b_1^{2^i+1} \\
&+& \left(\varphi_2+\varphi_4\right) \beta a_1^{2^i+1}b_1 + \frac{\left(\varphi_2+\varphi_4\right)^2\beta}{\varphi_4} a_1^{2^i}b_1^2
\end{eqnarray*}
\item 
\begin{eqnarray*}
	&& \mathrm{Det}(B_1) a_{24} \\
	&=& \frac{\left(\varphi_1+\varphi_2\right)^2\beta}{\varphi_4}a_1^{2^{i+1}+1}+\left(\varphi_1+\varphi_4\right)\beta a_1^{2^i+1}b_1^{2^i} +\frac{\left(\varphi_1+\varphi_4\right)^2\beta}{\varphi_4}a_1b_1^{2^{i+1}} \\
&+& \frac{\varphi_3\left(\varphi_2+\varphi_4\right)\beta}{\varphi_4} a_1^{2^{i+1}}b_1 + \frac{\left(\varphi_1+\varphi_4\right)\left(\varphi_2+\varphi_4\right)\beta}{\varphi_4} a_1^{2^i}b_1^{2^i+1}.
\end{eqnarray*}
\end{enumerate}
Hence, it is clear that 
\begin{eqnarray*}
	&& \begin{bmatrix}
		b_{21}b_{14}a_{11}+b_{21}b_{12}a_{13}+b_{22}b_{13}a_{11}+b_{22}b_{11}a_{13}, & b_{21}b_{14}a_{12}+b_{21}b_{12}a_{14}+b_{22}b_{13}a_{12}+b_{22}b_{11}a_{14} \\
		b_{23}b_{14}a_{11}+b_{23}b_{12}a_{13}+b_{24}b_{13}a_{11}+b_{24}b_{11}a_{13}, & b_{23}b_{14}a_{12}+b_{23}b_{12}a_{14}+b_{24}b_{13}a_{12}+b_{24}b_{11}a_{14}
	\end{bmatrix} \\
	&& =  \mathrm{Det}(B_1)   \begin{bmatrix}
		a_{21}, & a_{22} \\
		a_{23}, & a_{24}
	\end{bmatrix}. 
\end{eqnarray*} 
and Eq. (\ref{matrix1}) holds.

As for (ii), we need to show that
\begin{equation}
\label{matrix2}
A_2A_1^{-1}B_1=B_2,
\end{equation}
whose proof can be obtained through just changing $a_1$ and $b_1$ in the proof of (\ref{matrix1}).

 Therefore, for any  $(a_1,b_1), (a_2,b_2)\in\gf_{q}^2\backslash\{(0,0)\}$ satisfying $B_{V_i,(a_1,b_1)}(a_2,b_2)=(0,0)$, $\mathrm{Im}_{V_i,(a_1,b_1)}=\mathrm{Im}_{V_i,(a_2,b_2)}$ holds and by Lemma \ref{quadratic_boomerang}, we know that the boomerang uniformity of $V_i$ is $4$.

\section{Conclusions}
\label{Conclusions}

In this paper, we construct permutations with boomerang uniformity $4$ from the butterfly structure. For the open butterfly, there seems to be no any experiment results about permutation with boomerang uniformity over $\gf_{q}^2$ with $q=2^3$ by MAGMA. As for the closed butterfly,  we provide a condition such that $V_i$ is a permutation over $\gf_{q}^2$ with  boomerang uniformity $4$, where  with $q=2^n$. Moreover, the experiment results by MAGMA over  $\gf_{q}^2$ with $q=2^3, 2^5$ show that our condition in Theorem \ref{main_theorem} such that $V_i$ is a permutation with boomerang uniformity $4$ over $\gf_{q}^2$ is also necessary. We give the following conjecture here and invite interested readers to solve it. 

\begin{Conj}
	Let $q=2^n$ with $n$ odd, $\gcd(i,n)=1$ and $V_i := (R_i(x,y), R_i(y,x))$ with $R_i(x,y)=(x+\alpha y)^{2^i+1}+\beta y^{2^i+1}$. 
	Then if $V_i$ is a permutation over $\gf_q^2$ with boomerang uniformity  $4$, we have $(\alpha,\beta)\in\Gamma$ defined by (\ref{Gamma}).
\end{Conj}

Based on a private communication with M. Calderini, we look into the relation between the proposed quadratic permutation $V_i$ and the Gold function $x^{2^{2i}+1}$ characterized in \cite{boura2018boomerang}.  Let $q=2^n$ with $n$ odd. Let $L_1(x)=Ax^q+Bx, L_2(x)=Cx^q+Dx$ with $A,B,C,D\in\gf_{q}$ be permutations over $\gf_{q^2}$ and $$G_i(x)=L_2(L_1(x)^{2^i+I}),$$ where $\gcd(i,n)=1$, $I=1$ for $i$ even and $I=q$ for $i$ odd. For $i$ even, it is easy to obtain that $$G_i(x)=\epsilon_1x^{q\cdot(2^i+1)}+ \epsilon_2 x^{q\cdot 2^i+1} +  \epsilon_3x^{2^i+q} +  \epsilon_4 x^{2^i+1},$$
where 
 \begin{equation}
 \label{G}
\left\{
\begin{array}{lr}
\epsilon_1 = A^{2^{i}+1}D+B^{2^i+1}C\\
\epsilon_2= A^{2^i}BD+AB^{2^i}C \\
\epsilon_3= AB^{2^i}D+A^{2^i}BC\\
\epsilon_4=A^{2^i+1}C+B^{2^i+1}D.
\end{array}
\right.
\end{equation}
 Moreover,  after computing directly, we have $\varphi_2^{2^i}=\varphi_1\varphi_4^{2^i-1}$ and $\varphi_4\neq 0$, where $\varphi_j$ and $\epsilon_j$ are defined by (\ref{varphi}) and (\ref{G}), respectively. The case $i$ odd is similar. 
Furthermore, experimental results on $n=3, 5$ indicate that there exist some $A,B,C,D\in\gf_{q}$ such that $G_i(x)=V_i(x)$ for any $(\alpha,\beta)\in\Gamma$. In other word, the quadratic permutation $V_i$ seems to be affine equivalent to the Gold function. Therefore, unluckily, we may not obtain new permutations with boomerang uniformity $4$ from the butterfly structure. 


\bibliographystyle{IEEEtran}
\bibliography{ref}

\end{document}